\documentclass[11pt]{article}

\usepackage{comment}
\usepackage{amsmath, amsthm, amssymb, MnSymbol, fullpage, stmaryrd, xspace}
\usepackage[usenames,dvipsnames]{color}
\usepackage{enumerate}
\usepackage{umoline}
\usepackage[normalem]{ulem}

\newtheorem{theorem}[]{Theorem}[section]
\newtheorem{claim}[theorem]{Claim}
\newtheorem{corollary}[theorem]{Corollary}

\newtheorem{remark}[theorem]{Remark}

\newtheorem{example}[theorem]{Example}
\newtheorem*{example*}{Example}

\newcommand{\cc}[1]{\mathcal{#1}} 

\newcommand{\SF}[1]{\mathsf{#1}}

\newcommand{\Z}{\mathbb{Z}}

\newcommand{\Id}{\mathsf{Id}}
\newcommand{\Neg}{\mathsf{Neg}}
\newcommand{\true}{\mathrm{T}}
\newcommand{\false}{\mathrm{F}}
\renewcommand{\P}{\mathrm{P}}
\newcommand{\NP}{\mathrm{NP}}

\newcommand{\comp}{\mathrm{comp}}
\newcommand{\yes}{\mbox{\rmfamily\textsc{yes}}}
\newcommand{\no}{\mbox{\rmfamily\textsc{no}}}
\newcommand{\exc}{\mbox{\rmfamily\textsc{exception}}}

\newcommand{\ALPHA}[1]{\llbracket #1 \rrbracket}
\newcommand {\ie} {\textit{i.e.}\xspace}
\newcommand {\st} {\textit{s.t.}\xspace}

\newcommand{\E}[1]{\mbox{\rmfamily\textsc{E--#1}}}

\newcommand{\X}[1]{\mbox{\rmfamily\textsc{X--#1}}}

\newcommand{\SP}{\mbox{\rmfamily\textsc{S}}}
\newcommand{\HH}[1]{\mbox{\rmfamily\textsc{H--#1}}}

\newcommand{\CSP}{\mathrm{CSP}}

\newcommand{\KWEIGHT}{k\SF{WEIGHT}}
\newcommand{\HCSP}{\HH{$\CSP$}}
\newcommand{\GROUPEQ}{\SF{GROUPEQ}}
\newcommand{\ENT}{\SF{ENTRIES}}
\newcommand{\GROUP}{\SF{GROUP}}

\newcommand{\LINEQ}{\SF{LINEQ}}

\newcommand{\SAT}[1]{\mathsf{#1SAT}}
\newcommand{\UG}[1]{\mathsf{UG}[#1]}
\newcommand{\COL}[1]{\mathsf{#1COL}}
\newcommand{\CISO}[1]{#1\mathsf{CLQ\textsc{--}ISO}}
\newcommand{\CLQ}[1]{#1\mathsf{CLQ}}
\newcommand{\EQ}[1]{\mathsf{EQ}_{#1}}

\newcommand{\TWOCOL}{\COL{2}}
\newcommand{\newsat}{\mathsf{MONSAT}}

\newcommand{\CR}{\{\cc{R}\}}
\newcommand{\CV}{\{\cc{V}\}}
\newcommand{\onlyC}{\emptyset}

\newcommand{\V}{\mathsf{V}}

\newcommand{\prom}{\mathsf{PROM}}
\newcommand{\rf}{\mathsf{RF}}

\definecolor{edcolor}{rgb}{0,0.8,0.3}

\begin{document}

\title{On the complexity of trial and error for constraint satisfaction 
problems\footnote{A preliminary report on this work appeared in ICALP 14 as
\cite{conf_version}.}
}

\author{
G\'abor Ivanyos\thanks{Institute for Computer Science and Control, Hungarian Academy of Sciences, 
Budapest, Hungary 
({\tt Gabor.Ivanyos@sztaki.mta.hu}).} 
\and 
Raghav Kulkarni\thanks{Centre for Quantum Technologies,
 National University of Singapore, Singapore 117543 ({\tt kulraghav@gmail.com}).}
\and
 Youming Qiao\thanks{Centre for Quantum Software and Information, 
 University of Technology Sydney, Australia 
 ({\tt jimmyqiao86@gmail.com}).}
\and
 Miklos Santha\thanks{IRIF, CNRS, Universit\'e Paris Diderot, 75205 Paris, 
 France;  and 
Centre for Quantum Technologies, National University of Singapore, 
Singapore 117543 ({\tt santha@irif.fr}).}
\and 
Aarthi Sundaram\thanks{Centre for Quantum Technologies,
 National University of Singapore, Singapore 117543 
 ({\tt aarthims@u.nus.edu}).}
}
\date{}
\maketitle

\begin{abstract}
In 2013  Bei, Chen and Zhang introduced a trial and error model of computing, and applied to some constraint satisfaction problems. In this model the input is hidden by an oracle which, for a candidate assignment, reveals some information about a violated constraint if the assignment is not satisfying. 
In this paper we initiate a {\em systematic} study of constraint satisfaction problems in the trial and error model. To achieve this, we first adopt a formal framework for CSPs, and based on this framework we define several types of revealing oracles. Our main contribution is to develop a \emph{transfer theorem} for each type of the revealing oracle, under a broad class of parameters. To any hidden CSP with a specific type of revealing oracle, the transfer theorem associates another, potentially harder CSP in the normal setting, such that their complexities are polynomial time equivalent. This in principle transfers the study of a large class of hidden CSPs, possibly with a promise on the instances, to the study of CSPs in the normal setting. We then apply the transfer theorems to get polynomial-time algorithms or hardness results for hidden CSPs, including satisfaction problems, monotone graph properties, isomorphism problems, and the exact version of the Unique Games problem. Most of the proofs of these results are short and straightforward, which exhibits the power of the transfer theorems. 
\end{abstract}


 


\thispagestyle{empty}

\newpage

\setcounter{page}{0}
\pagenumbering{arabic}

\section{Introduction}
\label{sec:intro}

%

In \cite{BCZ13}, Bei, Chen and Zhang proposed a \emph{trial and error} model to study algorithmic problems when some input information is lacking. As argued in their paper, the lack of input information can happen when we have only limited knowledge
of, and access to the problem. They also described several realistic scenarios where the inputs are actually unknown. 
%
Then, they formalized this methodology in the complexity-theoretic setting, and 
proposed a trial and error model 
for constraint satisfaction problems. They further applied this idea to investigate the information needed to solve linear programming in \cite{BCZ13b}, and to study information diffusion in a social network in \cite{BCD+13}.


As mentioned, in \cite{BCZ13} the authors focused on the hidden versions of some specific
constraint satisfaction problems ($\HH{$\CSP$}$s), whose instances could only be accessed via a \emph{revealing} oracle.
An algorithm in this setting interacts with this revealing oracle to get information about the input instance. 
Each time, the algorithm proposes a candidate solution, a {\em trial},
and the validity of this trial is checked by the oracle. If the trial succeeds, the algorithm is notified that the proposed trial is already a solution. Otherwise, the algorithm obtains as an {\em
error}, a violation of some property
corresponding to the instance. The algorithm aims to make effective use
of these errors to propose new trials, and the goal is to minimize the number of trials while keeping in mind the cost for proposing new trials.
When the $\CSP$ is already difficult, a computation oracle that solves the 
original problem might be allowed. Its use is justified as we are interested in 
the \emph{extra difficulty} caused by the lack of information. Bei, Chen and 
Zhang considered several natural $\CSP$s in the trial and error setting, 
including $\SAT{}$, Stable Matching, Graph Isomorphism and Group Isomorphism. 
While the former two problems in the hidden setting are shown to be of the same 
difficulty as in the normal one, the last two cases have substantially increased 
complexities in the unknown-input model. They also studied more problems, as 
well as various aspects of this model, like the query complexity.

In this paper, following \cite{BCZ13}, we initiate a {\em systematic} study of 
the constraint satisfaction problems in the trial and error model. To achieve 
this, we first adopt a formal framework for $\CSP$s, and based on this 
framework we define three types of revealing oracles. 
This framework also helps to clarify and
enrich the model of \cite{BCZ13}. 
Let us make a quick remark that, our $\CSP$ model has a couple of features that may not 
be quite standard. We will mention some of these in the following, and discuss 
these in detail in  Section~\ref{subsec:std}. 

Our main contribution is to develop a 
\emph{transfer theorem} for each type of the revealing oracle, under a broad 
class of parameters. For any hidden $\CSP$ with a specific type of revealing 
oracle, the transfer theorem associates another $\CSP$ in the normal (unhidden) 
setting, such that their difficulties are roughly the same. This in principle 
transfers the study of hidden $\CSP$s to the study of $\CSP$s in the normal 
setting. We also apply transfer theorems to get results for concrete $\CSP$s, 
including some problems considered in \cite{BCZ13}, for which we usually get 
much shorter and easier proofs. 

\paragraph{The framework for $\CSP$s, and hidden $\CSP$s.} To state our results 
we describe informally the framework of $\CSP$s. A CSP $\SP$ is defined by a 
finite alphabet $\ALPHA{w}=\{0, 1, \dots, w-1\}$ and by
$\cc{R}=\{R_1, \dots, R_s\}$, a set of
relations over $\ALPHA{w}$ of some fixed arity $q$.
For a set of variables $\cc{V}=\{x_1, \dots, x_\ell\}$, an instance of 
$\SP$ is a set of constraints $\cc{C} = \{C_1, \dots, C_m\}$, where $ C_j = R(x_{j_1}, \ldots , x_{j_q})$ for some
relation $R \in \cc{R}$ and some $q$-tuple of variables. An assignment $a\in\ALPHA{w}^{\ell}$ satisfies 
$\cc{C}$ if it satisfies every constraint in it.

\begin{example}\label{ex:1sat} 
{\rm $\SAT{1}$: Here $w=2, \; q = 1,$ and 
$\cc{R} = \{\Id , \Neg \},$ where $\Id=\{1\}$ is the identity relation, and $\Neg =\{0\}$ is its complement. 
Thus a constraint is a literal $x_i$ or $\bar{x}_i$, and an instance is just a
collection of literals. In case of {\rm $\SAT{3}$} the parameters are $w=2, \; q = 3$ and $|\cc{R}| = 8$. 
We will keep for further illustrations $\SAT{1}$ which is a problem in polynomial time. {\rm $\SAT{3}$} would be a less
illustrative example since the standard problem is already $\NP$-complete}.

\end{example}

To allow for more versatility, we may only be interested in those assignments
satisfying certain additional conditions that cannot be (easily) expressed in the 
framework of constraint satisfaction problems.
This case happens, say when we look
for permutations in isomorphism problems, or when we view monotone graph 
properties
as $\CSP$ problems in Section~\ref{sec:onlyC}.
To cover these cases, our model will also include
a subset $W\subseteq \ALPHA{w}^\ell$ as a parameter
and we will look for satisfying assignments from $W$, whose elements
will be refereed to as {\em admissible assignments}. That these
admissible assignments can play
a notable role (as in Section~\ref{sec:onlyC}) is a first feature that may be
somewhat surprising.


Recall that in the hidden setting, the algorithm interacts with some revealing oracle by 
repeatedly proposing assignments.  If the proposed assignment 
is not satisfying then the revealing oracle discloses
certain information about some violated constraint. 
This can be in principle an index of 
such a constraint, (the index of) the relation in it, the indices of the variables where this relation is applied,
or any subset of the above. Here we will require that the oracle always
reveals the index of 
a violated constraint
from $\cc{C}$. 
To characterize the choices for the additional information,
for any subset $\cc{U} \subseteq \{\cc{R}, \cc{V}\}$ we say that an oracle is $\cc{U}$-{\em revealing}
if it also  gives out
the information corresponding to $\cc{U}$. For a $\CSP$ problem $\SP$ we use $\HH{$\SP$}_{\cc{U}}$ to denote the corresponding hidden problem in the trial and error model with $\cc{U}$-revealing oracle.

\vskip 0.3em
\noindent {\bf Example~\ref{ex:1sat} continued.} 
Let us suppose that we present an assignment $a \in \{0,1\}^{\ell}$ for an
instance of  the hidden version $\HH{$\SAT{1}$}_{\cc{U}}$ of {\rm $\SAT{1}$ to the $\cc{U}$-revealing oracle. If\ \ $\cc{U} = \{\cc{V}\}$
and the oracle reveals $j$ and $i$  respectively for the violated constraint and the variable in it
then we learn that the $j$th literal is $x_i$ if $a_i = 0$, and $\bar{x}_i$ otherwise.
If $\cc{U} = \{\cc{R}\}$ and say the oracle reveals $j$ and $\Id$ then we learn that the $j$th literal is positive.
If $\cc{U} = \emptyset$ and the oracle reveals $j$ then we only learn that the $j$th literal is either a positive literal corresponding to one of
the indices where $a$ is $0$, or a negative literal corresponding to an index where $a$ is $1$.

\vskip 0.3em



In order to explain the transfer theorem and motivate the operations which 
create richer CSPs, we  first make a simple observation that $\HH{S}_{\{\cc{R}, 
\cc{V}\}}$  and $\SP$
are polynomial time equivalent, when the relations of $\SP$ are in $\P$. Indeed, an algorithm for $\HH{S}_{\{\cc{R}, \cc{V}\}}$
can solve $\SP$, as the answers of the oracle can be given by directly 
checking if the
proposed assignment is satisfying.
In the other direction, 
we repeatedly submit assignments to the oracle. 
The answer of the oracle fully reveals a (violated) constraint. Given some subset
of constraints we already know, to find a new constraint, we submit an assignment which satisfies all the known constraints.
Such an assignment can be found by the algorithm for $\SP$.

With a weaker oracle this procedure clearly does not work and to compensate, we need stronger $\CSP$s. 
In the case of $\{\cc{V}\}$-revealing oracles an answer helps us
exclude, for the specified clause, all those relations which were satisfied 
at the specified indices of the proposed assignment,
but keep as possibilities all the relations which were
violated 
 at those indices. Therefore, to find out more information about the input, we would like to find a
satisfying assignment for a CSP instance whose corresponding
constraint is the union of the violated 
relations. 
This naturally brings us to consider the 
constraint satisfaction problem $\bigcup\SP$, the
{\em closure by union}
of $\SP$. The relations for $\bigcup\SP$
are from $\bigcup \cc{R}$, the \emph{closure by union} of $\cc{R}$, 
which contains
relations by taking union over any subset of $\cc{R}$.

The situation with the $\{\cc{R}\}$-revealing oracle is analogous, but here we have to compensate, in the 
stronger CSP, for the lack of revealed information about the variable indices.
For a relation $R$ and $q$-tuple of indices $(j_1, \ldots , j_q)$, we define the $\ell$-ary relation
$R^{(j_1, \ldots , j_q)} = \{a \in W ~:~ (a_{j_1}, \ldots , a_{j_q}) \in R\},$ 
and for a set $I$ of $q$-tuples of indices, we set
$R^I = \bigcup_{(j_1, \ldots , j_q) \in I} R^{(j_1, \ldots , j_q)} $.
The {\em arity extension} of
$\SP$ is the constraint satisfaction problem \E{S}  whose relations are from arity extension
\E{$\cc{R}$} $= \bigcup_{I}  \{R^I ~:~  R \in \cc{R} \}$ of $\cc{R}$.
Note that the arity extension produces relations whose 
arities are as large as the assignment length. This requires us to consider 
$\CSP$s where the arities of relations can be functions in e.g. the assignment 
length. While 
such $\CSP$s include some natural instances like systems of linear equalities, 
this feature may also be unfamiliar to some readers. 

The transfer theorem first says that with $\bigcup\SP$ (resp. \E{S}) we can compensate the information hidden
by  a $\{\cc{V}\}$-revealing (resp. $\{\cc{R}\}$-revealing) oracle, that is we 
can solve 
$\HH{S}_{\{\cc{V}\}}$ (resp. $\HH{S}_{\{\cc{R}\}}$). In fact, with $\bigcup \E{S}$ we can solve
$\HH{S}_{\emptyset}$. Moreover, perhaps more surprisingly, it says that these statements also hold in the reverse direction:
if we can solve the hidden CSP, we can also solve the corresponding extended CSP.

\vskip 0.3em
\noindent {\bf Transfer Theorem (informal statement)} {\em 
Let $\SP$ be a CSP whose parameters are ``reasonable" and whose relations are in $\P$. 
Then for any promise $W$ on the assignments, the complexities of the following problems are polynomial time equivalent:
(a) $\HH{S}_{\{\cc{V}\}}$ and $\bigcup \SP$, 
(b) $\HH{S}_{\{\cc{R}\}}$  and  \E{S},  
(c) $\HH{S}_{\emptyset}$ and $\bigcup  \E{S}$. }
\vskip 0.3em
The precise dependence on the parameters can be found in the theorems of 
Section~\ref{sec:reductions}. Corollary~\ref{equivalence} highlights the 
conditions for polynomial equivalence.

\vskip 0.3em
\noindent {\bf Example~\ref{ex:1sat} continued.} 
Since $\bigcup \{\Id, \Neg\} = \{ \emptyset, \Id, \Neg, \{0,1\}\}$, 
$\bigcup\SAT{1}$ has only the two trivial (always false or always true) relations in addition to the relations in $\SAT{1}$.
Therefore it can be solved in polynomial time, and by the the Transfer Theorem 
$\HH{$\SAT{1}$}_{\{\cc{V}\}}$ is also in $\P$.
On the other hand, for any index set $I\subseteq [\ell]$, $\Id^I$ is a disjunct of positive literals with variables from $I$,
and similarly $\Neg^I$ is a disjunct of negative literals with variables from $I$. Thus $\E{$\SAT{1}$}$ 
includes $\newsat$, which consists of those instances of $\SAT{}$ where in each clause either every variable is positive,
or every variable is negated. The problem $\newsat$ is $\NP$-hard by Schaefer's characterization~\cite{Sch78},
and therefore the Transfer Theorem implies that 
$\HH{$\SAT{1}$}_{\{\cc{R}\}}$ and $\HH{$\SAT{1}$}_{\emptyset}$ are also $\NP$-hard.
\vskip 0.3em

In a further generalization, we will 
also consider 
CSPs and $\HCSP$s whose 
instances satisfy some property. 
One such property can 
be  {\em repetition freeness}  meaning
that the constraints of an instance are pairwise distinct.  
The promise $\HCSP$s could also be a suitable framework
for discussing certain  graph problems on
special classes of graphs. For a promise $\prom$ on instances of $\SP$ we denote by $\SP^{\prom}$ the promise problem
whose instances are instances of $\SP$ satisfying $\prom$. The problem
$\HH{S}_{\{\cc{U}\}}^{\prom}$ is defined in an analogous way from $\HH{S}_{\{\cc{U}\}}$.


It turns out that we can generalize the Transfer Theorem for CSPs with promises on the instances.
We describe this in broad lines for the case of ${\{\cc{V}\}}$-revealing oracles.
Given a promise $\prom$ on $\SP$, the corresponding promise $\bigcup \prom$ for $\bigcup\SP$ 
is defined in a natural way.
We say that a $\bigcup\SP$-instance $\cc{C}'$  {\em includes} an $\SP$-instance $\cc{C}$ 
if for every $j\in[m]$, the constraint $C'_j$ in $\cc{C}'$ and the constraint $C_j$ in $\cc{C}$ 
are defined on the same variables, and seen as relations, $C_j \subseteq C'_j$.
Then $\bigcup \prom$ is the set of instances $\cc{C'}$ of $\bigcup \SP$  which 
include some $\cc{C} \in \prom$. The concept of an algorithm {\em solving} $\bigcup\SP^{\bigcup \prom}$
has to be relaxed: 
while we search for a satisfying assignment for those instances which include a satisfiable instance of $\prom$,
when this is not the case, the algorithm can abort even if the instance is satisfiable.
With this we have:

\vskip 0.3em
\noindent {\bf Transfer Theorem for promise problems (informal statement)} {\em 
Let $\SP$ be a constraint satisfaction problem with promise $\prom$. Then the complexities of 
$\HH{S}_{\{\cc{V}\}}^{\prom}$ and $\bigcup S^{\bigcup \prom}$ are polynomial time equivalent
when the parameters are ``reasonable" and the relations of $\SP$  are in $\P$}.

\vskip 0.3em
\noindent {\bf Example~\ref{ex:1sat} continued.} Let $\rf$ denote the property of being repetition free,
in the case of $\SAT{1}$ this just means that no literal can appear twice in the formula.
Then $\HH{$\SAT{1}$}_{\emptyset}^{\rf}$, hidden repetition-free $\SAT{1}$ with $\emptyset$-revealing oracle, 
is solved in polynomial time. 
To see this we first consider $\X{$\SAT{1}$}$, the constraint satisfaction problem whose relations are
all $\ell$-ary extensions of $\Id$ and $\Neg$. (See Section~\ref{sec:prelims} 
for a formal definition.)
It is quite easy to see that hidden $\SAT{1}$
with $\emptyset$-revealing oracle is essentially the same problem as hidden $\X{$\SAT{1}$}$
with $\{\cc{V}\}$-revealing oracle. Therefore, by the Transfer Theorem we are concerned with
$\bigcup \X{$\SAT{1}$}$ with promise $\bigcup \rf$. The instances satisfying the promise are
$\{C_1, \dots, C_m\}$, where $C_j$ is a disjunction of literals such that there exist distinct 
literals $z_1, \dots, z_m$, with $z_j\in C_j$. It turns out that these specific instances of $\SAT{}$ can be solved
in polynomial time. The basic idea is that 
we can apply a maximum matching algorithm, and only output a solution if we can select $m$ pairwise 
different variables $x_{i_1}, \dots, x_{i_m}$ such that either $x_{i_j}$ or $\overline{x}_{i_j}$ is in $C_j$. 

\vskip 0.3em

\noindent{\bf Applications of transfer theorems.}
Since $\NP$-hard problems obviously remain $\NP$-hard in the hidden setting (without access to an $\NP$ oracle), we investigate the complexity of various polynomial-time solvable $\CSP$s.
We first apply the Transfer Theorem when there is no promise on the instances. We 
categorize the hidden CSPs depending on the type of the revealing oracle.

With constraint and variable index revealing oracles, we obtain results on several interesting families of $\CSP$s including the exact-Unique Games Problem (cf. Section~\ref{sec:CV}), equality to a
member of a fixed class of graphs.
Interestingly, certain
$\CSP$s, like $\SAT{2}$ and the exact-Unique Game problem on alphabet size
$2$ remain in $P$,
while
some other $\CSP$s 
like the exact-Unique Game problem on alphabet size $\geq 3$, 
and equality to some specific graph, such as $k$-cliques,
become $\NP$-hard in this hidden input setting.
The latter problem is just the Graph Isomorphism problem considered in 
\cite[Theorem 13]{BCZ12}, whose proof, with the help of the Transfer Theorem, becomes very simple. 

With constraint and relation index revealing oracles, we show that if the arity 
and the alphabet size are constant, any $\CSP$ satisfying certain modest 
requirement becomes $\NP$-hard. To be specific, we require that 
that for every element $\alpha$ of the alphabet, the collection 
$\cc{R}$ contains a relation which is {\em violated} by
the tuple $(\alpha,\ldots,\alpha)$. This can be justified
by observing that otherwise it would be easy to find a satisfying 
assignment for any instance using $O(w)$ trials.

We then study various monotone graph 
properties like Spanning Tree, Cycle Cover, 
etc.. We define a general framework to represent variants of monotone graph 
property problems 
as $\HH{$\CSP$}$s. Since in this framework only one relation is present, the 
relation index is not a concern. We deal with the constraint index revealing 
oracle, which is equivalent to the constraint and relation index revealing oracle 
in this case, and prove that the problems become $\NP$-hard. 
This framework also naturally 
extends to directed graphs.


Finally, we investigate hidden CSPs with promises on the instances. We first 
consider the repetition freeness promise, as exhibited by the $\SAT{1}$ example as 
above. Though the hidden repetition free $\SAT{1}$ problem becomes solvable in 
polynomial time, $\SAT{2}$ is still $\NP$-hard. The group isomorphism problem can 
also be cast in this framework, and we give a simplified proof of \cite[Theorem 
11]{BCZ13}: to compute an explicit isomorphism of the hidden group with $\Z_p$ is 
$\NP$-hard. 

\paragraph{Comparisons with \cite{BCZ13}.} We now compare our framework and 
results with those in \cite{BCZ13} explicitly. Recall that we defined three 
revealing oracles, $\emptyset$-, $\{\cc{V}\}$-, and $\{\cc{R}\}$-revealing 
oracles. The 
$\emptyset$-revealing oracle was the original setting discussed in \cite{BCZ13}. 
The 
$\{\cc{V}\}$- and $\{\cc{R}\}$-revealing oracles are new, so are the results about 
specific 
$\CSP$s in the setting of these two oracles. For the $\emptyset$-revealing oracle, 
both \cite{BCZ13} and this paper discussed SAT and
Group Isomorphism. Isomorphisms of two graphs and isomorphisms
of a graph with a clique were studied in the report
\cite{BCZ12}. 
Here we prove a hardness result for the latter problem.
The paper \cite{BCZ13} further considered several other problems including 
stable matching, and Nash Equilibrium. On the other hand the monotone graph 
properties (Section~\ref{sec:onlyC}) and certain promise problems 
(Section~\ref{sec:promise}) are studied only in this paper.

Bei, Chen and Zhang also gave bounds 
on the {\em trial complexity} of some of the 
problems considered in \cite{BCZ13}, including stable matching and SAT. 
(The trial complexity measures the number of
oracle calls to solve them in the hidden model).
Although the algorithms
outlined in the proofs for our transfer theorems provide generic upper bounds 
on the trial complexity, giving tighter bounds would be beyond the focus of 
the present paper.






\vskip 0.3em
\noindent{\bf Organization.}
In Section~\ref{sec:prelims} we formally describe the model of $\CSP$s, and hidden $\CSP$s. 
In Section~\ref{sec:reductions}, the transfer theorems are stated and proved. 
Section~\ref{sec:CV}, \ref{sec:CR}, and \ref{sec:onlyC} contain the 
applications of the main theorems in the
case of $\CV$-revealing oracle, $\CR$-revealing oracle, and monotone graph 
properties,  
respectively. Finally in Section~\ref{sec:promise} we present the results for 
hidden promise CSPs.

\section{Preliminaries}\label{sec:prelims}

\subsection{The model of constraint satisfaction problems}\label{subsec:csp_model}

For a positive integer $k$, let $[k]$ denote the set $\{1, \ldots, k\}$, and let 
$\ALPHA{k}=\{0, 1, \dots, k-1\}$.
A {\em constraint satisfaction problem}, ($\CSP$) $\SP$, is specified by its set of parameters and its type, both
defined for every positive integer $n$.

\paragraph{The parameters.} The {\em parameters} are the {alphabet size} $w(n)$, 
the assignment length $\ell(n)$,
the set of (admissible) {assignments} $W(n) \subseteq \ALPHA{w(n)}^{\ell(n)}$, the {arity} $q(n)$,
and the {number of relations}
$s(n)$. 
To simplify notations, we often omit $n$ from the parameters, and just write 
$w, \ell, W, q$ and $s$. 

We suppose that the parameters, as functions of 
$n$, can be computed in time polynomial in $n$. 
In many cases (like in classical 
CSPs) $n$ coincides with $\ell$,
the assignment length but for e.g. (monotone) graph properties
the $n$ is the number of vertices while the assignment length is
$\binom{n}{2}$, the number of possible edges.

\paragraph{The type.} For a sequence $J=(j_1,\ldots,j_q)$ of $q$ distinct indices 
we
denote $W_J$ the projection of $W$ to the coordinates from $J$:
$W_J=\{(v_1,\ldots,v_q)\in [w]^q:\exists (w_1,\ldots,w_\ell)\in W\mbox{~with~}
w_{j_i}=v_i\}$. We suppose that $W_J$ does not depend on the choice of $J$,
that is, for every $J$ consisting of $q$ distinct indices we have
$W_J=W_q := \{u \in \ALPHA{w}^q ~:~ uv \in W ~ \mbox{{\rm for some}} ~ v \in 
\ALPHA{w}^{\ell - q}\}$. 
This condition holds 
trivially for most cases, and for other cases (e.g. $\CSP$s related to graphs), 
holds due certain symmetry condition there (e.g. graph properties are invariant 
for 
isomorphic graphs). 
A $q$-ary {\em relation} is $R \subseteq W_q$.
For $b$ in $W_q$, if $b \in R$, we sometimes write $R(b)=\true$, 
and similarly for $b \not \in R$ we write $R(b)=\false$. 
The {\em type} of $\SP$ is a set of $q$-ary relations $\cc{R}_n = \{R_1, \ldots, R_{s}\}$, where
$R_k \subseteq  W_q$, for every $k \in [s]$. As for the parameters, we usually just write $\cc{R}$.
Observe that the type of a $\CSP$ automatically defines among its parameters the 
arity and the number of relations.

We assume that the alphabet set and the relation set have succinct 
representations. Specifically, every letter and relation can be encoded by strings 
over $\{0, 1\}$ of length polynomial in $n$, and given such a string, we can 
decide whether it is a valid letter or relation efficiently. We also 
suppose the
existence of Turing
machines that, given $n$, words $b\in \ALPHA{w}^q$, $v\in \ALPHA{w}^\ell$,
$k \in [s]$, decide whether $v\in W$, whether $b\in W_q$ and compute
$R_k(b)$ if $b\in W_q$. 
We further introduce the following notations. For a relation $R$ let $\comp(R)$ be 
the time complexity of
deciding the membership of a tuple in $R$, and for a set of relations $\cc{R}$ let 
$\comp(\cc{R})$ be $\max_{R \in \cc{R}} \comp(R)$. We denote
by $\dim(\cc{R})$, the {\em dimension} of $\cc{R}$, which is
 the maximum of the integers $d$ such that there exists $R_1,\ldots,R_d\in\cc{R}$ with
$R_1\subsetneq R_2\subsetneq \ldots\subsetneq R_d$. In other words,
$\dim(\cc{R})$ is the length 
of the longest chain
of relations (for inclusion) in $\cc{R}$.

\paragraph{The instances.} We set $[\ell]^{(q)} = \{(j_1, \ldots, j_q) \in 
[\ell]^q ~:~ |\{j_1, \ldots, j_q\}| = q\}$, that is $[\ell]^{(q)}$ denotes 
the set of distinct $q$-tuples from $[\ell]$.
An {\em instance} of $\SP$ is given by
a set 
of $m$ 
constraints $\cc{C} = \{C_1, \ldots , C_m\}$ 
over a set $\cc{V}=\{x_1, \dots, x_{\ell}\}$ of {variables},
where the {\em constraint} $C_j$ is $R_{k_j}(x_{j_1}, \ldots , x_{j_q})$ for some $k_j \in [s]$ and $(j_1, \ldots , j_q) \in [\ell]^{(q)}$.
We say that an assignment $a \in W$ {satisfies} $C_j = R_{k_j}(x_{j_1}, \ldots , x_{j_q})$ if
$R_{k_j}(a_{j_1} , \ldots , a_{j_q}) = \true$. An assignment {\em satisfies} $\cc{C}$ if it satisfies all its constraints.
The {\em size} of an instance is $n + m(\log s + q \log \ell) + \ell \log w$ which includes the length of the description
of $\cc{C}$ and the length of the assignments. In all our applications the 
instance size will be polynomial in $n$.
A {\em solution} of $\cc{C}$ is a satisfying assignment if there exists any, 
and $\no$ otherwise.

Note that the size of an instance does not count 
the descriptions of the relations and the admissible assignments. This is because 
the latter information is thought of as the meta data of a $\CSP$, and is known to 
the algorithm. 

\paragraph{Operations creating new $\CSP$s from old $\CSP$s.} We also introduce 
two new operations which create richer sets of relations from a relation set.
For a given $\CSP$ $\SP$, 
these richer sets of relations derived from the type of $\SP$, will be the types of harder $\CSP$s which turn out to be
equivalent to
various hidden variants of $\SP$. The first operation is standard.
We denote
by $\bigcup \cc{R}$ the closure of $\cc{R}$ by the union operation, that is $\bigcup \cc{R} = \{ \bigcup_{R \in \cc{R}'} R ~:~ \cc{R}' \subseteq \cc{R}\}$.
We define the ({\em closure by}) {\em union} of $\SP$ as the constraint satisfaction problem
$\bigcup \SP$ whose type is $\bigcup \cc{R}$, and 
whose other parameters are the same as those of $\SP$.
We assume that a
relation $R$ in $\bigcup \cc{R}$ is represented a list of indices for relations 
from
$\cc{R}$ whose union is $R$.
We remark that $\dim (\bigcup\cc{R})\leq |\cc{R}|$.

For a relation $R \in \cc{R}$ and for $(j_1, \ldots , j_q) \in [\ell]^{(q)}$, we define the $\ell$-ary relation
$R^{(j_1, \ldots , j_q)} = \{a \in W ~:~ (a_{j_1}, \ldots , a_{j_q}) \in R\},$ 
and the {\em arity extension} of $\cc{R}$, as 
\X{$\cc{R}$} $= \{ R^{(j_1, \ldots , j_q)} ~:~ R \in \cc{R} \mbox{{\rm ~ and }} (j_1, \ldots , j_q) \in [\ell]^{(q)} \}.$
The set \X{$\cc{R}$} contains the natural extensions 
of relations in $\cc{R}$ to $\ell$-ary relations, 
where the extensions of the same relation are distinguished according to the 
choice of the $q$-tuple
where the assignment is evaluated.
The {\em arity extension} of
$\SP$ is the constraint satisfaction problem \X{S} whose type is \X{$\cc{R}$},
and whose other parameters are otherwise the same as those of $\SP$.
We assume that a
relation $R$ in \X{$\cc{R}$}  is represented by the index of a relation
in $\cc{R}$ and a sequence from $[\ell]^{(q)}$.

The combination of these two operations applied to $\cc{R}$ gives $\bigcup$ \X{$\cc{R}$},
the union of the arity extension of $\cc{R}$, which contains arbitrary unions of arity extended relations.
It will be useful to consider also restricting the unions to extensions coming from the same base relation.
For  $I \subseteq  [\ell]^{(q)}$, we set
$R^I = \bigcup_{(j_1, \ldots , j_q) \in I} R^{(j_1, \ldots , j_q)} $, 
and we define \E{$\cc{R}$} $= \{ R^I ~:~ R \in \cc{R} \mbox{{\rm ~ and  }} I \subseteq  [\ell]^{(q)} \}.$
Relations from $\bigcup$ \X{$\cc{R}$} are assumed to be represented as
a sequence of pairs, each consisting of an index of a relation in $\cc{R}$
and an element of $[\ell]^{(q)}$.

The {\em restricted union of arity extension} of
$\SP$ is the constraint satisfaction problem \E{S} whose type is \E{$\cc{R}$},
and whose other parameters are otherwise the same as those of $\SP$.
Observe that \E{$\cc{R}$} $\subseteq \bigcup$ \X{$\cc{R}$} $=  \bigcup$ \E{$\cc{R}$}. 




\subsection{Hidden $\CSP$ in the trial and error model} Suppose that we want to 
solve a $\CSP$ problem $\SP$ whose
parameters and type are known to us, but for the instance $\cc{C}$, we are explicitly given only 
$n$ and the number of constraints $m$. 
The instance is otherwise specified by a {\em revealing} oracle
$\V$ for $\cc{C}$ which can be used by an algorithm to receive
information about the constraints in $\cc{C}$.
The algorithm can propose  $a \in W$ to the oracle which
is conceived as its guess for a satisfying assignment. If $a$ indeed satisfies 
$\cc{C}$ then $\V$ answers $\yes$.
Otherwise there exists some violated constraint $C_j=R_{k_j}(x_{j_1}, \dots, 
x_{j_q})$, and the oracle has to reveal
some information about that. We will require that the oracle always reveals 
$j$, the index of the constraint $C_j$ in $\cc{C}$,
but in addition, it can also make further disclosures. It can be
the relation $R_{k_j}$ in  $\cc{R}$ 
(or equivalently its index ${k_j}$); it can also be
$(j_1, \ldots, j_q)$, the $q$-tuple of indices of the ordered variables
$x_{j_1}, \ldots , x_{j_q}$ in $\cc{V}$; or both of these. To characterize the choices for the additional information,
for any subset $\cc{U} \subseteq \{\cc{R}, \cc{V}\}$, we require that a $\cc{U}$-{\em revealing} oracle $\V_{\cc{U}}$ give out
the information corresponding to $\{\cc{C}\} \bigcup  \cc{U} \subseteq \{ \cc{C}, \cc{R}, \cc{V}\}$.
Thus for example a $\emptyset$-revealing oracle $\V_{\emptyset}$ reveals the index $j$ of some violated constraint but nothing else,
whereas a $\cc{V}$-revealing oracle $\V_{\{\cc{V}\}}$ also reveals the indices $(j_1, \ldots, j_q)$ of the variables of the relation in the clause $C_j$,
but not the name of the relation.

Analogously, for every $\CSP$ $\SP$, and
for every $\cc{U} \subseteq \{ \cc{R}, \cc{V}\}$, we define the {\em hidden constraint satisfaction problem} ($\HCSP$)
with $\cc{U}$-{\em revealing oracle}
$\HH{S}_{\cc{U}}$ whose parameters and type are those of $\SP$, but whose instances are specified by a $\cc{U}$-revealing oracle.
An algorithm {\em solves} the problem $\HH{$\SP$}_{\cc{U}}$ if for all $n,m$, for every instance $\cc{C}$ for $\SP$, specified by
any $\cc{U}$-revealing oracle for $\cc{C}$, it
outputs a satisfying assignment if there exists any, and $\no$ otherwise.
The complexity of an algorithm for $\HH{S}_{\cc{U}}$ is the number of steps in the worst case 
over all inputs and all $\cc{U}$-revealing oracles, where a query to the oracle is 
counted as one step.

\subsection{Discussion on the model}\label{subsec:std}

Our $\CSP$ model as described in Section~\ref{subsec:csp_model} includes the usual 
model of $\CSP$s, with a few additional features. 
Recall that we defined the alphabet size $w(n)$, the arity $q(n)$, and the number 
of relations $s(n)$ to be functions in $n$.\footnote{That the assignment length 
$\ell$ is a function of $n$ is, as far as we can see, quite standard in the 
literature. } The 
most common case is of course when 
these functions are just constant. We allow them to be functions for the following 
two 
reasons. Firstly, this allows us to include a couple of natural problems, like 
systems of linear equations over a finite field ($q(n)$ and $s(n)$ are not 
constant; see Claim~\ref{claim:hard}) and hyperplane non-cover ($w(n)$ is not 
constant; see 
Section~\ref{sec:CV}). Secondly, and more importantly, the arity extension  
produces relations whose arities $q$ are the same as the input length $\ell$, 
which then depends on $n$. 
This also makes the number of relations dependent on $n$, and 
the union operator may introduce even more new relations.

We see from above that, allowing $w$, $q$, and $s$ to be functions in $n$ is not 
only flexible, but also necessary for our purposes. This may cause some problems 
though, if we do not pose any constraint on such functions. We remedy 
these  
as 
follows, as already described in  
in Section~\ref{subsec:csp_model}. Firstly, we 
assume that all these functions can be computed in time polynomial in $n$. 
Secondly, we assume that the alphabet set and the relation set have succinct 
representations, and that
membership of tuples in 
every relation $R$ can be 
decided  
in time $\comp(R)$. 

While the above measures may look somewhat inconvenient, in all concrete $\CSP$s 
considered in this paper, they are satisfied in a straightforward manner. In most 
cases, $w(n)$ is polynomial in $n$; the only exception is hyperplane non-cover in 
Section~\ref{sec:CV}. The non-constant arity situation is mostly caused by arity 
extensions, in which case the new arity is just the assignment length $\ell$, 
and the number of new relations can be computed efficiently easily. 
The only problem with non-constant arities, not caused by arity extensions, is 
systems of linear questions as studied in Claim~\ref{claim:hard}. The relations 
created by the union operation or the arity 
extension operation have natural succinct representations, and the number of such 
relations can be easily computed. 
Furthermore, starting with a set of 
relations $\cal{R}$ with $\comp(\cc{R})$, whether an admissible assignment $v\in 
W$ satisfies a relation in $\bigcup$ \X{$\cc{R}$} involving 
$I\subseteq 
[\ell]^{(q)}$ can be computed in time 
$O(|I|\cdot |\cc{R}|\cdot \comp(\cc{R}))$.

Last but not least, the set of admissible assignments $W$ can also play 
a crucial  
role. For example in Section~\ref{sec:onlyC} we discuss monotone graph 
properties and the admissible assignments are minimal graphs satisfying a 
particular graph property. This set $W$ may cause similar 
problems if we do not pose any conditions on it, so we require that the membership 
of 
$W$ can be computed. For all cases in this paper this is satisfied trivially. 




\section{Transfer Theorems for Hidden CSPs}
\label{sec:reductions}

In this section we state and prove our transfer theorems between $\HCSP$s and $\CSP$s with extended types.

\begin{theorem}
 \label{cv_reduction}
\label{thm:UandH}

(a) If  $\bigcup \SP$ is solvable in time $T$ then $\HH{S}_{\{\cc{V}\}}$  is solvable in time
$O((T + s \times \comp(\cc{R})) \times m \times \min\{\dim (\bigcup \cc{R}), |W_q|\})$. \\
(b) If $\HH{S}_{\{\cc{V}\}}$ is solvable in time $T$ then $\bigcup \SP$  is solvable in time $O(T \times m \times \comp(\bigcup \cc{R}))$.
\end{theorem}

We stress that in the theorem above instances of 
$\HH{S}_{\{\cc{V}\}}$ consisting of $m$ relations correspond to
instances of $\bigcup \SP$ also consisting of $m$ relations.

\begin{proof}
We first prove (a). Let $\cc{A}$ be an algorithm which solves $\bigcup \SP$ in time $T$. We define an algorithm $\cc{B}$ for $\HH{S}_{\{\cc{V}\}}$. The algorithm
will repeatedly call $\cc{A}$, until it finds a satisfying assignment or reaches the conclusion $\no$.
Here is a brief and somewhat informal description.
\begin{itemize}
\item[]
{\it Initialization.}
Set $A_1=A_2=\ldots=A_m=\emptyset$ and let $I_1,\ldots,I_m$ 
be arbitrary $q$-tuples of variable indices. 
\item[]
{\it Loop.} (Repeat the following steps until termination.)
\begin{itemize}

\item
Set $C_j$ to be the union of the relations from $\cc{R}$ 
violated by every tuple in $A_j$ ($j=1,\ldots,m$). 

\item
Call algorithm $\cc{A}$ for $\{C_1,\ldots,C_m\}$. Return $\no$ if
$\cc{A}$ returned $\no$. Otherwise let $a=(a_1,\ldots,a_\ell)$ be the 
assignment returned by $\cc{A}$.

\item
Call oracle $\V$ with assignment $a$. Return $a$ if
$\V$ accepted it. Otherwise let $j$ and $I=(j_1,\ldots,j_q)$ be the constraint
index resp.~the variable index array revealed by $\V$.

\item
Set $I_j=I$, add $a'=(a_{j_1},\ldots,a_{j_q})$ to
$A_j$ and continue loop.
\end{itemize}
\end{itemize}
In the following more detailed description, we indicate the actual
repetition number in upper indices. The instance $\cc{C}^t =\{C_1^t, \ldots, C_m^t\}$
of the $t$th call of $\cc{A}$ is defined as
$
C_j^t = \bigcup_{R \in \cc{R}: R \cap A_j^t = \emptyset} R (x_{j_1^t}, \ldots , x_{j_q^t})
$
where $A_j^t \subseteq W_q$ and $I_j=(j_1^t, \ldots , j_q^t)  \in [\ell]^{(q)}$, for $j \in [m],$
are determined successively by $\cc{B}$. The set $A_j^t$ reflects the algorithm's knowledge after $t$ steps:
it 
 contains those $q$-tuples which, at that instant, are known to be violating the $j$th constraint.
Initially $A_j^1 = \emptyset$
and $(j_1^1, \ldots , j_q^1)$ is arbitrary. If the output of $\cc{A}$ for $\cc{C}^t$ is $\no$ then $\cc{B}$ outputs $\no$.
If the output of $\cc{A}$ for $\cc{C}^t$ is $a \in W$ then $\cc{B}$ submits $a$ to the $\{\cc{V}\}$-revealing oracle $\V$.
If $\V$ answers $\yes$ then $\cc{B}$ outputs $a$. If the oracle does not find $a$ satisfying, and reveals
$j$ and $(j_1, \ldots , j_q)$ about the violated constraint, then $\cc{B}$ does not change $A_i^t$ and $(i_1^t, \ldots , i_q^t)$
for $i \neq j$, but
sets $A_j^{t+1} = A_j^{t} \bigcup \{ (a_{j_1} , \ldots , a_{j_q}) \}$,
and $(j_1^{t+1}, \ldots , j_q^{t+1}) = (j_1, \ldots , j_q)$. Observe that the $q$-tuple
for the $j$th constraint is changed at most once, the first time when the revealing oracle gives the index of the $j$th constraint.

To prove that the algorithm correctly solves $\HH{S}_{\{\cc{V}\}}$, let $\cc{C} = \{C_1, \ldots , C_m\} $ be an instance of
$\SP$ and let $\V$ be any
$\{\cc{V}\}$-revealing oracle for $\cc{C}$.
We have to show that if $\cc{B}$ answers $\no$ then $\cc{C}$ is unsatisfiable.
If $\cc{B}$ answers $\no$, then for some $t$, the $t$th call  of $\cc{A}$ resulted in output $\no$. 
By construction, $A_j^t$
and $(j_1^t, \ldots , j_q^t)$ are such
that for every $j \in [m]$,
if $R \cap A_j^t \neq \emptyset$ then 
the relation of $C_j$ can't be $R (x_{j_1^t}, \ldots , x_{j_q^t})$. Indeed, if 
$C_j = R (x_{j_1^t}, \ldots , x_{j_q^t})$ and
$ b \in R \cap A_j^t$  then
at the call when $b$ was added to $A_j^t$ the oracle's answer is incorrect. Therefore all possible remaining 
relations for $C_j$s are included
in $C_j^t$, and since $\cc{C}^t$ is unsatisfiable, so is $\cc{C}$.

For the complexity of the algorithm let us remark that if
for some $j$ and $t$, the constraint $C_j^t$ is the empty relation then $\cc{B}$ stops since $\cc{C}^t$ becomes unsatisfiable.
This happens in particular if
$A_j^t = W_q$. Since for every call to $\cc{A}$ one new element is added to one of the $A_j^t$ and at least one new relation in $\cc{R}$
is excluded from $C_j^t$, the number of calls is upper bounded by  $m \times \min\{\dim (\cc{R}), |W_q|\}.$ To compute a new constraint,
some number of relations in $\cc{R}$ have to be computed on a new argument, which can be done in time $s \times \comp(\cc{R})$.

We now prove (b). Let $\cc{A}$ be an algorithm which solves $\HH{S}_{\{\cc{V}\}}$ in time $T$.
Without loss of generality we suppose that $\cc{A}$ only outputs a satisfying assignment $a$ after submitting it to the verifying oracle.
We define an algorithm $\cc{B}$ for $\bigcup \SP$.
Let $\cc{C} =\{C_1, \ldots, C_m\}$ be an instance of $\bigcup \SP$ where
for $j \in [m]$,
$
C_j = \bigcup_{R \in \cc{R}_j} R (x_{j_1}, \ldots , x_{j_q}),
$
for some $\cc{R}_j \subseteq \cc{R}$ and $I_j=(j_1, \ldots , j_q)  \in [\ell]^{(q)}$.
The algorithm $\cc{B}$
runs $\cc{A}$, and outputs $\no$ whenever $\cc{A}$ outputs $\no$.
During $\cc{A}$'s run $\cc{B}$
simulates a ${\{\cc{V}\}}$-revealing oracle $\V$ for $\cc{A}$. 
Here is rather informal description of $\cc{B}$. 

\begin{itemize}
\item[]
{\it Initialization.}
Set $A_1=A_2=\ldots=A_m=\emptyset$ and run $\cc{A}$ until it
calls oracle $\V$ first time.
\item[]
{\it Loop.} (Repeat the following steps until termination.)
\begin{itemize}
\item
Let $a=(a_1,\ldots,a_\ell)$ be the assignment with which $\cc{A}$ calls
$\V$. Check if $a$ satisfies $\cc{C}$. Return $a$ if yes.

\item
Choose index $j$ such that $a$ violates $C_j$ and
add tuple
$(a_{j_1},\ldots,a_{j_q})$ to $A_j$. 

\item
Run subsequent steps of $\cc{A}$  
(with $j$ and $(j_1,\ldots,j_q)$ as revealed information)
until the next oracle call. 
\end{itemize}
\end{itemize}

Now we give more details about how $\cc{B}$ implements $\V$.
Simultaneously with $\V$'s description, for $t \geq 1$, we also specify instances $\cc{C}^t =\{C_1^t, \ldots, C_m^t\}$
of $\bigcup \SP$ which
will be used in the proof of correctness of the algorithm. For $j \in [m]$, the constraints of $\cc{C}^t $ are defined as
$
C_j^t = \bigcup_{R \in \cc{R}_j: R \cap A_j^t = \emptyset} R (x_{j_1}, \ldots , 
x_{j_q}),$
where the sets $A_j^t \subseteq W_q$ are determined by the result of the $t$th call to the oracle.
Initially $A_j^0 = \emptyset$.
For the $t$th request $a \in W$, the algorithm
$\cc{B}$ checks if $a$ satisfies $\cc{C}$. If it is the case then
$\V$ returns $\yes$ 
and $\cc{B}$ outputs $a$. Otherwise there exists $j \in [m]$
such that $a$ violates $C_j$, and the answer of the oracle is $j$ and $(j_1, \ldots , j_q)$
(where $j$ can be chosen arbitrarily among the violated constraints, if there are several).
Observe that this is a legitimate oracle for any instance of
$\HH{S}_{\{\cc{V}\}}$ whose $j$th constraint is arbitrarily chosen from $\cc{R}_j$.
We define $A_j^{t} = A_j^{t-1} \bigcup \{ (a_{j_1} , \ldots , a_{j_q}) \}$, and for
$i \neq j$ we set $A_i^{t} = A_i^{t-1}$.

To show the correctness of $\cc{B}$, we prove that whenever $\cc{A}$ outputs $\no$, the instance $\cc{C}$
is unsatisfiable. Let us suppose that $\cc{A}$ made $t$ queries before outputting $\no$. An algorithm
for $\HH{S}_{\{\cc{V}\}}$ can output $\no$ only if all possible instances of $\SP$ which are compatible
with the answers received from the oracle are unsatisfiable. In such an instance the 
relation of the $j$th constraint has necessarily
empty intersection with $A_j^t$, therefore we can deduce that the $\bigcup \SP$ instance $\cc{C}^t $ is
unsatisfiable. It also holds that $A_j^t \bigcap (\bigcup_{R \in \cc{R}_j}R) = \emptyset$ 
for every $j \in [m]$, since if $b \in A_j^t \bigcap (\bigcup_{R \in \cc{R}_j}R)$
then the request to the oracle that caused $b$ to be
 added to $A_j^t$ wouldn't violate the $j$th constraint.
Thus $\bigcup_{R \in \cc{R}_j}R \subseteq \bigcup_{R \in \cc{R}: R \cap A_j^t = \emptyset}R$, 
and $\cc{C}$ is unsatisfiable.

For the complexity analysis we observe that during the algorithm, for every query to the oracle and for every constraint, one relation in $\bigcup \cc{R}$ is evaluated.
\end{proof}

\begin{theorem}
 \label{cr_reduction}
(a) If  \E{S} is solvable in time $T$ then $\HH{S}_{\{\cc{R}\}}$  is solvable in time
$O( (T +  
|[\ell]^{(q)}|  \times \comp(\cc{R})  ) \times m \times 
|[\ell]^{(q)}|).$ \\
(b) If $\HH{S}_{\{\cc{R}\}}$ is solvable in time $T$ then \E{S}  is solvable in time $O(T \times m \times \comp(   \E{$\cc{R}$}        ))$.
\end{theorem}

Like above, instances of 
$\HH{S}_{\{\cc{R}\}}$ consisting of $m$ relations correspond to
instances of $\E{S}$ also consisting of $m$ relations.

\begin{proof}
The proof is similar to the proof of Theorem~\ref{cv_reduction}.
We first prove (a). Let $\cc{A}$ be an algorithm which solves $ \E{S}$ in time $T$. We define an algorithm $\cc{B}$ for $\HH{S}_{\{\cc{R}\}}$.
The algorithm will repeatedly call $\cc{A}$, until it finds a satisfying assignment or reaches the conclusion $\no$.
Since each constraint of $ \E{S}$ is an $\ell$-ary relation, we can identify it with 
the 
relation itself. Here is again a brief and informal description of $\cc{B}$.

\begin{itemize}
\item[]
%
{\it Initialization.}
Set $A_1=A_2=\ldots=A_m=\emptyset$ and
$C_1=C_2=\ldots=C_m=W$.

\item[]
{\it Loop.} (Repeat the following steps until termination.)
\begin{itemize}

\item
Call $\cc{A}$ for $(C_1,\ldots,C_m)$. Return $\no$ if
$\cc{A}$ returned $\no$. Otherwise let $a=(a_1,\ldots,a_\ell)$ be the assignment
output by $\cc{A}$.

\item
Call oracle $\V$ with assignment $a$. Return $a$ if
$\V$ accepted it. Otherwise let $j$ and $R$ be the constraint
index resp.~the relation revealed by $\V$.

\item
Add $a=(a_{j_1},\ldots,a_{j_\ell})$ to $A_j$,
let $I$ be the set of the 
variable index arrays $(j_1,\ldots,j_q)$ such that
$R^{(j_1,\ldots,j_q)}$ is violated by every tuple in $A_j$,
update $C_j$ to be $R^{I}$, the
union of the relations $R^{(j_1,\ldots,j_q)}$ over
$(j_1,\ldots,j_q)\in I$ and continue loop.
\end{itemize}
\end{itemize}

In the more detailed description and analysis we again apply 
repetition indices. For the first call $\cc{C}^1 =\{C_1^1, \ldots, C_m^1\}$ we set $C_j^1 = W$.
For $t>1$, the instance of the $t$th call will be defined recursively
via $A_j^t \subseteq W$ and $I_j^t \subseteq [\ell]^{(q)}$, for $j \in [m]$, where
initially we set $A_1^1 = \ldots = A_m^1 = \emptyset$ and $I_1^1 = \ldots = I_m^1 = [\ell]^{(q)}$.
Here the set $I_j^t$ reflects the algorithm's knowledge after $t$ steps:
it contains those $q$-tuples of indices which, at that instant, can still be the variable indices of the  
$j$th constraint.
If the output of $\cc{A}$ for $\cc{C}^{t-1}$
is $\no$ then $\cc{B}$ outputs $\no$.
If the output of $\cc{A}$ for $\cc{C}^{t-1}$ is $a \in W$ then $\cc{B}$ submits $a$ to the $\{\cc{R}\}$-revealing oracle $\V$.
If $\V$ answers $\yes$ then $\cc{B}$ outputs $a$. If the oracle does not find $a$ satisfying, and reveals
$j$ and $R \in \cc{R}$ about the violated constraint, then $\cc{B}$ does not change $A_i^{t-1}$, $I_i^{t-1}$ and
$C_i^{t-1}$
for $i \neq j$, but
sets $A_j^{t} = A_j^{t-1} \bigcup \{ a \}$
and $I_j^t = \{ (j_1, \ldots , j_q) ~:~  A_j^t  \bigcap R^{(j_1, \ldots , j_q)} = \emptyset \}$.
Finally we define
$C_j^t = R^{I_j^t}$.

To prove that the algorithm correctly solves $\HH{S}_{\{\cc{R}\}}$, let $\cc{C} = \{C_1, \ldots , C_m\} $ be an instance of
$\SP$ and let $\V$ be any
$\{\cc{R}\}$-revealing oracle for $\cc{C}$.
We have to show that if $\cc{B}$ answers $\no$ then $\cc{C}$ is unsatisfiable.
If $\cc{B}$ answers $\no$ then for some $t$, the $t$th call  of $\cc{A}$ resulted in output $\no$. 
We claim that for every constraint $C_j$ whose relation $R$ has been already revealed,
if $R^{(j_1, \ldots , j_q)} \cap A_j^t \neq \emptyset$ then $C_j$ can not be $R (x_{j_1}, \ldots , x_{j_q})$.
Indeed, if
$C_j = R^{(j_1, \ldots , j_q)} (x_{j_1}, \ldots , x_{j_q})$ and
$ a \in R \cap A_j^t$  then
at the call when $a$ was added to $A_j^t$ the oracle answer is incorrect. Therefore $C_j^t$ is the union,
over all still possible variable index $q$-tuples $(j_1, \ldots , j_q)$, of
 $R^{(j_1, \ldots , j_q)}$. Since $\cc{C}^t$ is unsatisfiable, so is $\cc{C}$.

For the complexity of the algorithm let us remark that if
for some $j$ and $t$, the constraint $C_j^t$ is the empty relation then $\cc{B}$ stops since $\cc{C}^t$ becomes unsatisfiable.
This happens in particular if
$I_j^t = \emptyset$. Since for every call to $\cc{A}$, for some $j$, the size of $I_j^t$ decreases by at least one, the total number of calls
is upper bounded by $m \times |[\ell]^{(q)}|$.To compute a new constraints,
at most $|[\ell]^{(q)}|$ relations from $\cc{R}$ evaluated in a new argument. Therefore the overall complexity is as claimed.

We now prove (b). Let $\cc{A}$ be an algorithm which solves $\HH{S}_{\{\cc{R}\}}$ in time $T$.
Without loss of generality we suppose that $\cc{A}$ only outputs a satisfying assignment $a$ after submitting it to the verifying oracle.
We define an algorithm $\cc{B}$ for $\E{S}$.
Let $\cc{C} =\{C_1, \ldots, C_m\}$ be an instance of $\E{S}$ where
for $j \in [m]$, we have $C_j = R_{k_j}^{I_j}$
for some $R_{k_j} \in \cc{R}$ and $I_j \subseteq [\ell]^{(q)}$.
The algorithm $\cc{B}$
runs $\cc{A}$, and outputs $\no$ whenever $\cc{A}$ outputs $\no$.
During $\cc{A}$'s run $\cc{B}$
simulates an ${\{\cc{R}\}}$-revealing oracle $\V$ for $\cc{A}$.
Here is an informal description of $\cc{B}$.
\begin{itemize}
\item[]
%
{\it Initialization.}
Set $A_1=A_2=\ldots=A_m=\emptyset$. Run $\cc{A}$ until it calls
$\V$ first time.

\item[]
{\it Loop.} (Repeat the following steps until termination.)
\begin{itemize}

\item
Let $a$ be the assignment that $\cc{A}$ submits to $\V$.
Return $a$ if it satisfies $\cc{C}$. 

\item
Choose and index $j$ such that $a$ violates $C_j$. 

\item
Run subsequent steps of $\cc{A}$ with revealed information
$j$ and $R_{k_j}$ until the next oracle call.
\end{itemize}
\end{itemize}
 
Now we give more details. Simultaneously with $\V$'s description, for $t \geq 1$, we also specify instances $\cc{C}^t =\{C_1^t, \ldots, C_m^t\}$
of $\E{S}$ which
will be used in the proof of correctness of the algorithm. 
Again we identify the $\ell$-ary constraints with their relations.
The constraints of $\cc{C}^t $ are set to be
$C_j^t = R_{k_j}^{I_j^t}$, where the sets $I_j^t \subseteq [\ell]^{(q)}$ are defined as
$I_j^t = \{ (j_1, \ldots , j_q) ~:~  A_j^t \bigcap R_{k_j}^{(j_1, \ldots , j_q)} = \emptyset \}$, and
the sets $A_j^t \subseteq W$ are determined by the result of the $t$th call to the oracle.
Initially $A_j^0 = \emptyset$.
For the $t$th request $a \in W$, the algorithm
$\cc{B}$ checks if $a$ satisfies $\cc{C}$. If it is the case then
$\V$ returns $a$ and $\cc{B}$ outputs $a$. Otherwise there exists $j \in [m]$
such that $a$ violates $C_j$, and the answer of the oracle is $j$ and $R_{k_j}$. 
Observe that this is a legitimate oracle for any instance of
$\HH{S}_{\{\cc{R}\}}$ whose $j$th constraint is arbitrarily chosen from $\{ R_{k_j}(x_{j_1}, \ldots, x_{j_q}) ~:~  (j_1, \ldots, j_q) \in I_j \}$.
We define $A_j^{t} = A_j^{t-1} \bigcup \{ a \}$, and for
$i \neq j$ we set $A_i^{t} = A_i^{t-1}$.

To show the correctness of $\cc{B}$, we prove that whenever $\cc{A}$ outputs $\no$, the instance $\cc{C}$
is unsatisfiable. Let us suppose that $\cc{A}$ made $t$ queries before outputting $\no$. An algorithm
for $\HH{S}_{\{\cc{R}\}}$ can output $\no$ only of all possible instances of $\SP$ which are compatible
with the answers received from the oracle are unsatisfiable. In such an instance the $j$th constraint has necessarily
empty intersection with $A_j^t$, therefore we can deduce that the $\E{S}$ instance $\cc{C}^t $ is
unsatisfiable. It also holds that $A_j^t \bigcap C_j = \emptyset$ for every $j \in [m]$, since if $a \in A_j^t \bigcap C_j$
then the request to the oracle that caused $a$ to be
 added to $A_j^t$ wouldn't violate the $j$th constraint.
Thus $C_j \subseteq C_j^t$, and $\cc{C}$ is unsatisfiable.

For the complexity analysis we just have to observe that during the algorithm, for every query to the oracle and for every constraint, one relation in
$ \E{$\cc{R}$} $ is evaluated.
\end{proof}


\begin{theorem}
\label{c_reduction}
 \label{thm:UXH}

(a) If $\bigcup \X{S}$ is solvable in time $T$ then $\HH{S}_{\emptyset}$  is solvable in time
$O(  (T + s \times  \frac{\ell !}{( \ell - q)!} \times \comp(\cc{R})  ) \times m \times \dim (\bigcup \mbox{\X{$\cc{R}$}}))$. \\
(b) If $\HH{S}_{\emptyset}$ is solvable in time $T$ then $\bigcup  \X{S}$  is solvable in time $O(T \times m \times \comp(\bigcup \mbox{\X{$\cc{R}$}}))$.
\end{theorem}


Again, instances of 
$\HH{S}_{\emptyset}$ consisting of $m$ relations correspond to
instances of $\bigcup \X(S)$ also consisting of $m$ relations.

\begin{proof} Apply Theorem~\ref{cv_reduction} to \X{S} and observe that $\HH{\X{S}}_{\{\cc{V}\}}$ and $\HH{S}_{\emptyset}$
are essentially the same in the sense that an algorithm solving one of the problems also solves the other one. Indeed, the
variable index disclosure of the $\{\cc{V}\}$-revealing oracle is pointless since the relations in \X{S} involve all variables.
Moreover, the map sending a constraint $R(x_{j_1}, \dots x_{j_q}) $ of {\SP} to the constraint
$R^{(j_1, \ldots , j_q)}(x_1, \ldots x_{\ell})$ of \X{S} is a
bijection which preserves satisfying assignments.
\end{proof}

\begin{corollary}
 \label{equivalence}
Let $\comp (\cc{R})$ be polynomial. Then the complexities of the following problems are polynomial time equivalent: 
(a) $\HH{S}_{\{\cc{V}\}}$ and $\bigcup \SP$ if the number of relations $s$ is constant, 
(b) $\HH{S}_{\{\cc{R}\}}$  and  \E{S} if the arity $q$ is constant, 
(c) $\HH{S}_{\emptyset}$ and $\bigcup  \X{S}$ if both $s$ and $q$ are constant.
\end{corollary}

The polynomial time equivalences of 
Theorems~\ref{cv_reduction},~\ref{cr_reduction},~\ref{c_reduction} and Corollary~\ref{equivalence} 
remain true when the algorithms have access to the same computational oracle. Therefore, 
we get generic easiness results for $\HCSP$s under an $\NP$ oracle. 

We also remark that the number of iterations in algorithms for
the transfer theorems give generic upper bounds on the {\em trial
complexity} of the hidden CSPs In the constraint index and variable 
revealing model this bound is $m\times \min\{\dim \bigcup\cc{R},|W_q|\}$, 
in the constraint index and relation
revealing model $m\times [\ell]^{(q)}$ oracle calls are sufficient, while
in the constraint index revealing model we obtain the bound $m\times \dim
\bigcup X-\cc{R}$.


\section{Constraint-index and Variables Revealing Oracle}
\label{sec:CV}

In this section, we present some applications of our transfer theorem 
when the index of the constraint and the variables participating 
in that constraint are revealed. We consider the following $\CSP$s.
In the descriptions below, unless explicitly specified, 
$W$ is the full domain $[w]^\ell$.

\begin{enumerate}
\item {Deltas on Triplets} ($\Delta$):
Formally, $w = 2, \; q = 3,$ and $\cc{R} = \{R_{abc} : \{0, 1\}^3 \to \{\true, \false\} \mid a,b,c \in \{0, 1\}\},$
where $R_{abc}(x, y, z) := (x=a) \wedge (y=b) \wedge (z=c).$ 
\item {Hyperplane Non-Cover} ($\SF{HYP\!-\!NC}$): Let $p$ be a prime, and $F_p$ 
be the field of size $p$. Denote $V = 
F_p^N$, and $S = \{\text{all hyperplanes in } F_p^N\}$. Informally, given a set of 
hyperplanes $S'\subseteq S$, the problem asks to decide if there exists $v\in 
F_p^N$ not covered by these hyperplanes. 
Formally, $\ell=1, \; q=1, \; w=p^N, \; W = V$  
and $\cc{R}_{S} = \{R_{H} \mid H \in S\}$ where $R_{H}(a)$ evaluates to $\true$ 
if and only if $a\notin H$.

\item {Arbitrary sets of binary relations on Boolean alphabet,} (in particular, $\SAT{2}$):  
Formally for $\SAT{2}$, we have $w=2, \; q = 2,$ and 
$\cc{R} = \{R_T,R_F, R_a, R_b, R_{\neg a}, R_{\neg b}, R_{a \vee b}, R_{a \vee \neg b},
R_{\neg a \vee b},  R_{\neg a \vee \neg b} \},$ where for $(\alpha, \beta) \in \{\true, \false\}^q,$ 
$R_T(\alpha, \beta) := \true, \; R_ F(\alpha, \beta) := \false, \; R_a(\alpha, \beta) := \alpha, 
\; R_b(\alpha, \beta) := \beta,$ \\ 
$R_{\neg a}(\alpha, \beta) := \neg \alpha, \; R_{\neg b}(\alpha, \beta) := \neg \beta, \; 
R_{a \vee b}(\alpha, \beta) := \alpha \vee \beta, \; R_{a \vee \neg b}(\alpha, \beta) := \alpha \vee \neg \beta,$ \\ 
$R_{\neg a \vee b}(\alpha, \beta) := \neg \alpha \vee \beta, \; 
R_{\neg a \vee \neg b}(\alpha, \beta) := \neg \alpha \vee \beta.$
\item {Exact-Unique Game Problem} ($\UG{k}$):
Given an undirected graph, $G = (V, E)$, and a permutation $\pi_e : \ALPHA{k} \to \ALPHA{k}$ 
for every edge $e \in E$, the goal is to decide if one can assign labels $\alpha_v \in \ALPHA{k}$ 
for every vertex $v \in V$ \st for every edge $e = \{u, v\} \in E$ with $u < v$ we have 
$\pi_e(\alpha_u) = \alpha_v.$ Formally: $w = k, \; q = 2$ and 
$\cc{R}= \{ \pi : \ALPHA{k} \to \ALPHA{k} \mid \pi \text{ is a permutation} \}.$ 
\item {$k$-Clique Isomorphism} ($\CISO{k}$): \label{itm:klq}
Given an undirected graph $G = (V, E),$ determine if there exists a permutation $\pi$ on $[n]$ \st
\begin{itemize} 
\item[(1)] $\forall (i, j) \in E, \; \;   R_{\leq k}(\pi(i), \pi(j))$;
\item[(2)] $\forall (i, j) \notin E, \; \; \neg R_{\leq k}(\pi(i), \pi(j)).$
\end{itemize}
Formally, $w = n, \; q = 2, \; \ell=n$, 
$W$ is the set of $n$-tuples of integers from $[n]$ which
define permutations on $[n]$, and 
$\;\; \cc{R} = \{R_{\leq k}, \neg R_{\leq k} \},$ where 
$R_{\leq k}(\alpha, \beta) := \true$ $\iff \alpha \leq k\; \&\; \beta \leq k.$ 
\item {Equality to some member in a fixed class of graphs} ($\EQ{\cc{K}}$):
For a fixed class $\cc{K}$ of graphs on $n$ vertices
variables, we denote by 
$\cc{P}_\cc{K} : \{0, 1\}^{n \choose 2} \to \{T, F\}$ 
the property of being equal to a graph from $\cc{K}$. We assume that
graphs are represented by tuples from $\{0,1\}^{n \choose 2}$.
Formally, 
$W = \cc{K}, \; w = 2, \; q = 1, \; \ell = {n \choose 2},$ and 
$\cc{R} = \{ \Id, \Neg \}.$ Here we assume that membership in $\cc{K}$ can
be tested in polynomial time. We will consider the following special cases:
\begin{itemize}
\item {Equality to $k$-Clique} ($\EQ{\CLQ{k}}$): 
Given a graph, decide if it is equal to a $k$-clique.
\item {Equality to Hamiltonian Cycle} ($\EQ{\SF{HAMC}}$):
Decide if $G$ is a cycle on all $n$ vertices.
\item {Equality to Spanning Tree} ($\EQ{\SF{ST}}$):
Given a graph, decide if it is a spanning tree.
\end{itemize}

\end{enumerate}

We have seen in the Introduction that the hidden version of 
$\SAT{1}$ in the constraint index revealing model is $\NP$-hard.
Here we will show that if the variables are also revealed, even the hidden
version of $\SAT{2}$ becomes solvable in polynomial time. Deltas on triplets 
will provide a simple example of ternary Boolean constraints 
for which the ``normal'' satisfaction problem can be solved in polynomial
time but the hidden version becomes $\NP$-hard. 
Hyperplane Non-Cover over the two-element field is equivalent to
a system of linear equations. Interestingly, its hidden version will
turn out to be $\NP$-hard. 
The unique game problem is a prominent CSP problem, whose approximate version has 
been studied intensively since \cite{Kho02}. It is known that the exact version is 
in $\P$ for any 
$k$, and we show that it is only easy in the trial and error model for $k=2$. 
Isomorphisms with cliques is a problem considered
in~\cite{BCZ12}. We included it to demonstrate how easy to prove hardness
of its hidden version based on the transfer theorem. Also note that
our hardness result is somewhat stronger than that of~\cite{BCZ12}
as the latter is proved for the constraint index revealing model while
here more information is revealed. A formally different,
although logically equivalent formulation of the same problem is
equality with a $k$-clique. The hardness result can be extended to
equalities with other distinguished graphs, like Hamiltonian circles.
However, equality with spanning trees will remain easy.

 \begin{theorem}
 \label{thm:cv-p}
 The following problems can be solved in polynomial time:
  (a) $\HH{$\SAT{2}$}_{\CV}$,
(b) $\HH{$\UG{2}$}_{\CV}$,  
(c)  $\HH{$\EQ{\SF{ST}}$}_{\CV}$.
 \end{theorem}

 \begin{theorem}
 \label{thm:cv-np}
 The following are $\NP$-hard:
  (a) $\HH{$\Delta$}_{\CV}$,
  (b) $\HH{$\SF{HYP\!-\!NC}$}_{\CV}$,
  (c) $\HH{$\UG{k}$}_{\CV}$ for $k \geq 3$,  
  (d)  $\HH{$\CISO{k}$}_{\CV}$ for $0.1 n \leq k \leq 0.9 n$,
  (e) $\HH{$\EQ{\CLQ{k}}$}_{\CV}$ for $0.1 n \leq k \leq 0.9 n$,
  (f)  $\HH{$\EQ{\SF{HAMC}}$}_{\CV}$.
 \end{theorem}

\begin{proof}[Proof of Theorem~\ref{thm:cv-p}] We show that the following problems 
in the hidden model with the constraint and variable index revealing oracle 
are solvable in polynomial time.

\begin{description}
\item[\textnormal{(a) Arbitrary binary Boolean relations ($\HH{$\SAT{2}$}_{\CV}$)}] \hfill \\
In the case of $\SAT{2}$, taking the union of any two relations in $\cc{R}_{\SAT{2}}$ is 
equivalent to the disjunction of the two boolean expressions the relations signify. For 
example, $R_{a} \bigcup R_{b} = R_{a \vee b}$ and the union remains in $\cc{R}_{\SAT{2}}$.
Hence, $\bigcup\SAT{2} = \SAT{2}$, which is in $\P.$ Therefore, from 
Theorem~\ref{thm:UandH}(a), $\HH{$\SAT{2}$}_{\{\cc{V}\}}$ is also in $\P.$ 

The above statement can be extended to an {\em arbitrary} set 
$\cc{R}'$ of binary relations as follows. Let $\cc{R}''$
stand for the set of {\em all} binary relations
in Boolean variables. We trivially have 
${\bigcup \cc{R}' \subseteq \cc{R}''}$, 
therefore an instance of $\HH{$\SAT{2}$}_{\{\cc{V}\}}$ can actually
be described by a conjunction of the form $\bigwedge_{k=1}^m R_k(x_{i_k},x_{j_k})$
where $R_k$ is a binary relation. Expressing each $R_k$ by a Boolean
formula in conjunctive normal form, we obtain an instance of $\SAT{2}$ 
consisting of $O(m)$ clauses, which can be solved in polynomial time.

\item[\textnormal{(b) Unique Games ($\HH{$\UG{2}$}_{\CV}$)}] \label{itm:UG2} \hfill \\
$\UG{2}$ is a $\CSP$ with $w = 2, \; q = 2,$ and
$\cc{R} = \{ \pi : \ALPHA{2} \to \ALPHA{2} \mid \pi  \text{ is a permutation}\}.$ The 
only permutations in $\cc{R}_{\UG{2}}$ enforce that either $\alpha_u = \alpha_v$ or 
$\alpha_u = \alpha_v \oplus 1$ for an edge $e = (u, v)$. Both of these relations 
can be represented as binary boolean relations. Hence, $\UG{2}$ is an instance of 
$\SAT{2}$ and from Theorem~\ref{thm:cv-p}$(a)$, $\HH{$\UG{2}$}_{\CV}$ is in $\P$.

\item[\textnormal{(c) Equality/Isomorphism to a member in a fixed class of 
graphs}] \hfill \\ \label{itm:graph-eq}
We define the $\HH{$\EQ{\cc{K}}$}_{\CV}$ problem in more detail. 
Let $\cc{K}$ be a class of graphs on $n$ vertices. 
We define $\cc{P}_\cc{K} : \{0, 1\}^{n \choose 2} \to \{T, F\}$ 
as the graph property of being equal to a graph from $\cc{K}$. 
Correspondingly, $W = \cc{K}$.

Formally, for $\cc{P}_\cc{K}$ we consider the $\CSP$ $\EQ{\cc{K}}$ with
$w = 2, \; q = 1, \; W = \{\alpha \in \{0, 1\}^{n \choose 2} \mid 
\alpha \in \cc{K}\},\; \ell = {n \choose 2},$ and $\cc{R} = \{ \Id, \neg \}.$ 
Given a graph instance $G = (V, E)$ in this model, the $n \choose 2$ constraints for
$G_2$ are such that $C_e = \Id(\alpha_e) \text{ for } e \in E$ 
and $C_e = \neg(\alpha_e)$ otherwise. 
This implies that $\bigcup \cc{R} = \{ \Id, \neg, \true \}$ and instances
of $\bigcup\EQ{\cc{K}}$ are parametrized with graphs (sets of edges) 
$E_1\subseteq E_2$. Here $E_2$ is the the set of unordered pairs $e$
for which the constraint is either $\Id$ or $\true$ while $E_1$ consists of
pairs for which the constraint is $\Id$. The $\bigcup\EQ{\cc{K}}$-problem 
becomes then:
given sets $E_1, E_2$ such that $E_1 \subseteq E_2$, does there exist a graph 
$G' =(V,E')\in \cc{K}$ 
such that $E_1 \subseteq E' \subseteq E_2$?

From Theorem~\ref{thm:UandH}, the complexity of $\HH{$\EQ{\cc{K}}$}$, can be analyzed by considering the complexity of
$\bigcup\EQ{\cc{K}}$. Below, we analyze the complexity of $\bigcup\EQ{\cc{K}}$ when $\cc{K}$ is the class of spanning trees 
on $n$ vertices.

\begin{remark}
\label{rem:sub}
For any $\cc{K}$, if we take $E_1 = \emptyset$, then solving $\EQ{\cc{K}}$ becomes equivalent to finding out if there exists $G \in \cc{K}$ which is a subgraph of $E_2$.
\end{remark}

\begin{remark}
\label{rem:iso}
Note that if we assume that $\cc{K}$ is the set of all graphs isomorphic to some $G_0$ and $E_1 = E_2$ as arbitrary graphs on $n$ vertices, then solving $\EQ{\cc{K}}$ becomes equivalent to finding out if $E_2$ is isomorphic to $G_0.$
\end{remark}

\emph{Proof for Equality to a Spanning Tree} ($\HH{$\EQ{\SF{ST}}$}_{\CV}$):
Here, $\cc{K}$ is the set of all possible spanning trees on $n$ vertices and $E_1$ without loss of generality is a 
forest $F$. $E_2$ is any arbitrary graph on $n$ vertices containing $E_1$. In this case, 
the $\bigcup{\EQ{\cc{K}}}$ problem becomes equivalent to finding a spanning tree on $E_2$ 
which also contains the forest $F$. This problem is in $P$ which implies that the
$\HH{$\EQ{\SF{ST}}$}_{\CV}$ problem is also in $P$.
\end{description}
This completes the proof for Theorem~\ref{thm:cv-p}.
\end{proof}

\begin{proof}[Proof of Theorem ~\ref{thm:cv-np}] We show that the following problems in the hidden model with the constraint and variable index revealing oracle are $\NP$-hard. Using Theorem~\ref{thm:UandH}, the complexity of each 
$\HH{$\SP$}_{\CV}$ is analyzed by considering the complexity of $\bigcup \SP$.

\begin{description}
\item[\textnormal{(a) Deltas on Triplets ($\HH{$\Delta$}_{\CV}$)}] \hfill \\
By definition, each relation in $\cc{R}_{\Delta}$ identifies a boolean string 
on $3$ variables. This implies that $\bigcup \cc{R}_{\Delta}$ forms the set 
of all Boolean predicates on $3$ variables. Thus, $\SAT{3}$ can be 
expressed as the $\bigcup \Delta$ problem. Hence, from 
Theorem~\ref{thm:UandH}(b), $\HH{$\Delta$}_{\CV}$ is $\NP$-hard.

%

\item[\textnormal{(b) Hyperplane Non-Cover ($\HH{$\SF{HYP\!-\!NC}$}_{\CV}$)}] 
\label{itm:subgrp} \hfill \\
The Hyperplane Non-Cover problem ($\SF{HYP\!-\!NC}$) is the solvability 
of homogeneous linear in-equations in $F_p^N$. The $\SF{HYP-NC}$ problem over
$Z_p^N$ for $p\geq 3$ 
includes the $\COL{3}$ problem and is already $\NP$-hard. To see this, let
$E$ be a graph on vertex set $[N]$ and consider the in-equations $x_i\neq x_j$ for
indices $i,j\in [N]$ such that 
$\{i,j\}\in E$. If $p=3$ that's all we need. Otherwise, add a variable $y$ 
together with the in-equations $y\neq 0$, $x_i\neq ky$ for $i\in [N]$ and $k\in
[p-3]$. 

Hence, it remains to consider the 
$\HH{$\SF{HYP\!-\!NC}$}_{\CV}$ problem over $F_2^N$, which we need to examine 
$\bigcup \SF{HYP\!-\!NC}$ by Theorem~\ref{thm:UandH}. In this setting, let $T$ be 
the set of all subspaces (not necessarily hyperplanes) of $F_2^N$. Then the set of 
constraints of
$\bigcup \SF{HYP\!-\!NC}$ consists of $\{R_P\mid P\in T\}$ where $R_P(a)$ 
evaluates to $\true$ if and only if $a\not\in P$. 
This problem 
is $\NP$-hard, as it includes non-covering by subspaces
 of codimension $2$
which encompasses the $\COL{4}$ problem. Hence, the former will be 
$\NP$-hard using Theorem~\ref{thm:UandH}.

\item[\textnormal{(c) Unique games ($\HH{$\UG{k}$}_{\CV}$ for $k \geq 3$)}] \hfill \\
$\UG{3}$ is a $\CSP$ with $w= 3, \; q = 2,$ and
$\cc{R} = \{ \pi : \ALPHA{3} \to \ALPHA{3} \mid \pi  \text{ is a permutation}\}.$
Let \[R^\circ := \mathop{\bigcup}_{\pi: (\forall i)(\pi(i) \neq i)} \pi .\] Note that 
$R^\circ \in \bigcup\cc{R}.$ Choosing $R^\circ$ as the constraint for every edge gives 
us the $\COL{3}$ problem. Hence, from Theorem~\ref{thm:UandH}(b),  
$\HH{$\UG{3}$}_{\CV}$ is $\NP$-hard.

\begin{remark}
Our proof method also shows that $\HH{$\UG{k}$}_{\CV}$ is $\NP$-hard for any $k > 2.$
\end{remark}

\item[\textnormal{(d) $k$-Clique Isomorphism ($\HH{$\CISO{k}$}_{\{\cc{V}\}}$) for
$0.1 n \leq k \leq 0.9 n$}] \label{itm:k-clique-iso} \hfill \\
Obviously, replacing the constraints 
of type ${R}_{\leq k}$ by $R_{\leq k}\cup \neg R_{\leq k}$ 
in an instance of $\HH{$\CISO{k}$}$
we obtain an instance of $\bigcup$$\HH{$\CISO{k}$}$.
This however just means omitting Constraints (1), and we obtain the 
$\CLQ{k}$ problem (deciding whether the graph contains a $k$-clique) which is $\NP$-hard. 
Hence from Theorem~\ref{thm:UandH}(b),  the $\HH{$\CISO{k}$}_{\{\cc{V}\}}$ problem is $\NP$-hard. 

\item[\textnormal{(e) Equality to a $k$-Clique ($\HH{$\EQ{\CLQ{k}}$}_{\CV}$) for
$0.1 n \leq k \leq 0.9 n$}] \hfill \\
We use the framework defined in the previous proof for the 
$\HH{$\EQ{\SF{ST}}$}_{\CV}$ problem.
As mentioned in Remark~\ref{rem:sub}, given a graph $E_2$, 
consider $\cc{K}$ to be the set of all possible $k$-cliques on $n$ vertices 
and $E_1=\emptyset$. In this setting, the 
$\bigcup\EQ{\cc{K}}$ problem 
is equivalent to finding a $k$-clique on $E_2$ which is $\NP$-hard.

\begin{remark}
The above proof could also serve as an alternate proof for Theorem~\ref{thm:cv-np}(d).
\end{remark}

\item[\textnormal{(f) Equality to a Hamiltonian Cycle ($\HH{$\EQ{\SF{HAMC}}$}_{\CV}$)}] \hfill \\
We use the framework defined in the previous proof for the $\HH{$\EQ{\SF{ST}}$}_{\CV}$ problem.
Here, $\cc{K}$ is the set of all possible Hamiltonian cycles on $n$ vertices and
$E_1=\emptyset$. For an arbitrary graph $E_2$.  
the $\bigcup{\EQ{\cc{K}}}$ problem parametrized by $E_1$ and $E_2$
becomes equivalent to deciding if $E_2$ has a Hamiltonian cycle, which is $\NP$-hard.
\end{description}
This completes the proof for Theorem~\ref{thm:cv-np}.
\end{proof}


\section{Constraint-index and Relation Revealing Oracle}
\label{sec:CR}

\begin{theorem} 
\label{c_boring}
Let $\SP$ be a $\CSP$ with constant arity $q$
and constant alphabet size $w$.
Assume that
for every $\alpha\in \ALPHA{w}$, there is  a non-empty relation $R_{\alpha} \in \cc{R}$ 
such that $(\alpha,\ldots,\alpha) \not\in R_{\alpha}$.
Then, $\HH{\SP}_{\CR}$ is $\NP$-hard.
\end{theorem}

\begin{proof}
We show that $\E{$\SP$}$ is $\NP$-hard. We will reduce to it the problem $\E{$\SAT{3}$}$ which consists of those
instances of $\SAT{3}$ where in each clause either every variable is positive, or every variable is negated. Restricting to these instances of $\SAT{3}$ is known as  
$\newsat$, whose $\NP$-completeness can be deduced, for example, from Schaefer's characterization~\cite{Sch78}.

We first extend the relations for $\SP$ as follows. Let $q'=(w-1)q+1$ and let $\cc{R}'\subseteq \ALPHA{w}^{q'}$ 
be the set of $q'$-ary relations that can be obtained
as an extension of an element of $\cc{R} \setminus \{\emptyset\}$ 
from any $q$ coordinates. 
Since $q$ and $w$ are constant, the cardinality of $\cc{R}'$ is also constant. We claim that 
$\bigcap_{R\in \cc{R}'} R=\emptyset$. Indeed, 
every $a\in \ALPHA{w}^{q'}$ has a sub-sequence $(\alpha,\ldots,\alpha)$ of
length $q$ for some $\alpha \in \ALPHA{w}$, therefore the extension of $R_{\alpha}$ from these $q$ coordinates
does not contain $a$.
Let $\{R^0,R^1,\ldots,R^h\}$ be a minimal subset of $\cc{R}'$
such that $\bigcap_{i=0}^h R^i=\emptyset$. Since the empty relation is not in
$\cc{R}'$, we have  $h\geq 1$.
Let us set $A^0 =  \bigcap_{i \neq 1} R^i$ and $A^1 =  \bigcap_{i \neq 0} R^i$. Then $A^0 \bigcap A^1 = \emptyset$,
and because of the minimality condition, $A^0 \neq \emptyset$ and $A^1 \neq \emptyset$. 

For a boolean variable $x$, we will use the notation $x^1 = x$ and $x^0 = {\bar x}$. The main idea of the proof
is to encode a boolean variable $x^1$ by the relation $A^1$ and $x^0$ by $A^0$. We think about the elements of $A^1$ as satisfying $x^1$,
and about the elements of $A^0$ as satisfying $x^0$. Then $x^1$ and $x^0$ can be both satisfied, but not simultaneously.

To implement the above idea, we extend the relations further, building on the above extension. We suppose without loss of generality that $\ell$ is a multiple of $q'$, and we set $\ell' = \ell/q'$.
Since $q'$ is constant, $\newsat$ on $\ell'$ variables is still NP-hard.
We take $ \ell' $ pairwise disjoint
blocks of size $q'$ of the index set $[\ell]$ and
on each block we consider relations $R^0,\ldots,R^h$.
We denote by $R^i_k$ the $\ell$-ary relation which
is obtained by extending $R^i$ from the $k$th block.
Observe that the relations $R^i_k$ are just extensions of elements of $\cc{R}$.

After these preparations, we are ready to present the construction. Let $K =  \bigwedge_{t=1}^{u} {K}_t$ be an instance of  $\E{$\SAT{3}$}$
in $\ell'$ variables, with each 3-clause of the form
$K_t = x_{i_{t,1}}^{b_t}  \vee x_{i_{t,2}}^{b_t} \vee x_{i_{t,3}}^{b_t},$
where $i_{t,1}, i_{t,2}, i_{t,3}$ are indices 
from $[\ell']$ and $b_t$ is either 0 or 1.
Then we map $K$ to the instance $\cc{C}$ whose constraints are
$$R_{i_{t,1}}^{b_t}  \cup R_{i_{t,2}}^{b_t} \cup R_{i_{t,3}}^{b_t},$$
for each $t\in [u]$, and
$$C^j_k=R^j_k,$$
for each  $k \in [\ell']$ 
and $j \in \{2,\ldots,h\}$.
This is an instance of $\E{$\SP$}$ since the three relations $R_{i_{t,1}}^{b_t}$,  $R_{i_{t,2}}^{b_t}$  and $R_{i_{t,3}}^{b_t}$
are the extensions of the same relation in $\cc{R}$. It is quite easy to see that
$K$ is satisfiable if and only if $\cc{C}$ is satisfiable.
Indeed, a satisfying assignment $a$ for the
$\cc{C}$ can be translated to a
satisfying assignment for $K$ by assigning
$0$ or $1$ to $x_k$ according to whether the $k$th block
of $a$ was in $A^0_k$ or $A^1_k$ (taking an arbitrary
value if it was in none of the two). Similarly, a satisfying assignment $b$ for $K$  can be translated to a
satisfying assignment $a$ for  $\cc{C}$ by picking any element of $A^{b_k}_k$ for the $k$th block of $a$.
\end{proof}

An immediate consequence is that under the same conditions $\HH{$\SP$}_{\onlyC}$ is $\NP$-hard too.  
For an application of this consequence, 
let $\LINEQ$ stand for the $\CSP$ in which that alphabet
is identified with a finite field $F$ and the $\ell$-ary constraints are linear
equations over $F$. 

\begin{claim}\label{claim:hard}
$\HH{$\LINEQ$}_{\onlyC}$  is $\NP$-hard.
\end{claim}

\begin{proof}
For each $i\in\ell$, we
pick two equations: $x_i=0$ and $x_i=1$. 
Observe that $x_i=0$ is the same
as $\{0\}^i$, the $\ell$-ary extension
of the unary relation $\{0\}$ on
the $i$th position and we have the same
if we replace $0$ by $1$. By the above observation, the 
$\HCSP$s built from relations of these type
are $\NP$-hard.
\end{proof}









\section{Monotone graph properties}
\label{sec:onlyC}

In this section, we consider monotone graph properties in the context of 
constraint index and relation revealing oracles. 

Let us 
first formulate 
monotone graph properties as $\CSP$s. Recall that
a {\em monotone graph property}
of an $n$-vertex graph is a monotone Boolean function $\cc{P}$ on $\binom{n}{2}$ 
variables indexed by $\{(i, j), 1\leq i<j\leq n\}$, 
invariant under the induced action of $S_n$ on $[n]$. The goal is to decide, given 
a graph $G = (V, E)$, whether $E \in 
\cc{P}$. 


For a monotone graph property $\cc{P}$, we turn it into a $\CSP$ problem as 
follows. Given 
a graph $G$, to decide whether $G$ satisfies the property $\cc{P}$, we are asked
to propose a minimal graph $H$ satisfying $\cc{P}$ such that $H$ is a subgraph of 
$G$. That is, $G$ is thought of as an instance, and $H$ is an assignment. 
If $H$ is a subgraph of $G$ then $H$ is considered as satisfying. If not, a 
violation is given by an edge in $H$ but not in $G$. That is, the set of 
constraints in 
the instance $G$ consists of negations of 
the variables corresponding to the
edges {\em not} in $G$. 
Formally, 
the $\CSP$ $\SP_{\cc{P}}$ associated with $\cc{P}$ has parameters
$w = 2, \; q = 1,$ $\ell = {\binom{n}{2}},$ $W_\cc{P}$
which consists of the (characteristic vectors of) the graphs (sets of edges),
minimal for inclusion satisfying $\cc{P}$,  and $\cc{R} = \{ \Neg \}$, 
where $\Neg$ is the negation function. The corresponding constraints are $\Neg(e)$ 
for every $e \notin E$. A graph $G=(V, E)$ yields an instance consisting of 
$\Neg(e)$ for $e\not\in G$. 
Then the goal is the same as to decide, given whether there exists an $A \in 
W_\cc{P}$ 
such that $A \subseteq E$. 
Notice that the graph property specific part of this model is {\em fully} left to
the specification of the set $W_\cc{P}$ of admissible assignments.

Let us examine such $\CSP$s in the trial and error setting. As there is only one 
relation $\cc{R}=\{\Neg\}$, the models revealing or not revealing the violated 
relation type are equivalent, so naturally the interesting setting is to work with 
the constraint index and 
relation 
revealing oracle. In particular, as already noted in the Introduction,
the case when the variable is also revealed is polynomial time equivalent
to the ``normal'' 
 problem. 
Also note that the union of arity 
extension and the restricted union of arity extension are the same. Then the 
arity 
extension is
\X{$\cc{R}$} $= \{ \Neg^e \mid e \in {n \choose 
2} \}$, where
$\Neg^e(\alpha_1, \ldots, \alpha_{n \choose 2}) = \neg \alpha_e$, and therefore 
$\bigcup$\X{$\SP_{\cc{P}}$}$=\{\vee_{e\in E'}\Neg^e\mid 
E'\subseteq 
\binom{n}{2}\}$ where $\vee$ denotes the or operator. That is, the relations in 
$\bigcup$\X{$\SP_{\cc{P}}$} are parametrized by sets $E'$ 
of possible
edges where the relation corresponding to $E'$ is equivalent to that
at least one edge of the proposed graph is not there in $E'$. 

As a consequence, the $\bigcup$\X{$\SP_{\cc{P}}$} problem becomes the following:
Given a graph $G = (V, E),$ and edge sets on $n$ vertices
$E_1, \ldots, E_m \subseteq {[n] \choose 2}$, 
does there exist an $A \in {W_\cc{P}}$ such that $ A \subseteq E$ and 
$A$ excludes at least one edge from each $E_i$? While every subset of the edge set
$\binom{n}{2}$ is in our disposal, in actual applications, the tricky part is to 
come up with appropriate subsets $E_1,\ldots,E_m$, which, together with the 
minimal 
instances of 
the graph property in question, yield that the resulting 
$\bigcup$\X{$\SP_{\cc{P}}$} problem is hard. This will be the main
content in  our proofs for the various parts of Theorem~\ref{thm:mono-graph}.

This framework naturally extends to directed and bipartite graphs as well
as to graphs with one or more designated
vertices. 
Monotone decreasing 
properties 
can be treated by replacing $\Neg$ with $\Id$, the identity function.

From Theorem~\ref{c_reduction}, the complexity of ${\HH{\SP}_{\cc{P}}}_{\emptyset}$
can be analyzed by considering the complexity of $\bigcup$\X{$\SP_{\cc{P}}$}.
We do this for the following graph properties (named after the minimal
satisfying graphs): 

\begin{enumerate}
\item {Spanning Tree ($\SF{ST}$):} the property of containing a spanning
tree, that is, being connected.
\item {Directed Spanning Tree ($\SF{DST}$):} the property of containing a 
directed spanning tree rooted at vertex (say) $1$.
such that all the edges of the spanning tree are directed towards the root.
\item {Undirected Cycle Cover ($\SF{UCC}$):} the property of containing an 
undirected cycle cover (union of vertex disjoint cycles such that every 
vertex belongs to some cycle).
\item {Directed Cycle Cover ($\SF{DCC}$):} the property of containing a 
directed cycle cover (union of vertex disjoint directed cycles such 
that every vertex belongs to some cycle).
\item {Bipartite Perfect Matching ($\SF{BPM}$):} the property of having a 
perfect matching in a bipartite graph.
\item {Directed Path ($\SF{DPATH}$):} the property of containing a directed 
path between two specified vertices $s$ and $t$.
\item {Undirected Path ($\SF{UPATH}$):} the property of containing an 
undirected path between two specified vertices $s$ and $t$.
\end{enumerate}

Connectedness of the given graph 
and connectivity of two designated vertices, as 
well as
 their directed versions, belong to the simplest well-known monotone graph 
 properties that are 
decidable in polynomial time. Having perfect matchings is perhaps the most famous property 
in $\P$ for bipartite graphs. We show that these problems become $\NP$-hard
in the hidden setting. We included the cycle cover problems because we prove
hardness of having perfect matchings through hardness of
$\HH{$\SF{DCC}$}_\emptyset$. The
hidden version of having a fixed subgraph, e.g., a clique of constant size 
is in $\P$ because there are only polynomially minimal satisfying graphs and they
can be efficiently listed. Unfortunately we are not aware of any monotone
property which remains efficiently decidable in the hidden setting for a less trivial
reason. 

 \begin{theorem}
 \label{thm:mono-graph}
 The following problems are $\NP$-hard:
  (1) $\HH{$\SF{ST}$}_\emptyset$,
(2) $\HH{$\SF{DST}$}_\emptyset$,  
(3)  $\HH{$\SF{UCC}$}_\emptyset$,
 (4) $\HH{$\SF{DCC}$}_\emptyset$,
 (5)  $\HH{$\SF{BPM}$}_\emptyset$, 
  (6) $\HH{$\SF{DPATH}$}_\emptyset$, 
(7)  $\HH{$\SF{UPATH}$}_\emptyset$.
 \end{theorem}
 
\begin{proof} We show that the following problems in the hidden model with the 
constraint index 
revealing oracle are $\NP$-hard. In each case, we construct an instance of the 
$\bigcup$\X{$\SP_{\cc{P}}$} \st it becomes equivalent to a known $\NP$-hard 
problem and 
using Theorem~\ref{c_reduction}(b) we can conclude that the hidden version, 
$\HH{\SP}_{\cc{P}}$, 
is $\NP$-hard.

\begin{description}
\item[\textnormal{(1) Spanning Tree ($\HH{$\SF{ST}$}_\emptyset$)}] \hfill \\
When $\cc{P}$ is {\em connectedness}, $W_\cc{P}$ is the set of {\em Spanning 
Trees} on $n$-vertices. 

Given $G = (V, E),$ for every vertex $v \in V,$ we consider ${n-1} \choose 3$ edge 
sets 
$E_{vijk }$ where 
\[E_{vijk} := \{\{v, i\},  \{v, j\}, \{v, k\} \} \quad 1 \leq i < j < k \leq n.\] 
With this choice of $E_{vijk}$s the $\bigcup$\X{$\SF{ST}$} problem
asks if there exists a spanning tree in $G$ which avoids at least one edge
from each $E_{vijk}$. This is equivalent to
 that every vertex $v$ is incident to at most two edges of the spanning tree
$A$. Spanning trees with this property are
just Hamiltonian paths in $G$. 

Thus, the $\bigcup$\X{$\SF{ST}$} problem is equivalent to asking if $G$ contains 
a Hamiltonian path \ie the $\SF{HAM\!-\!PATH}$ problem in $G$. Hence, the 
$\NP$-hard 
$\SF{HAM\!-\!PATH}$ in $G$ problem reduces to the $\HH{$\SF{ST}$}_\emptyset$ 
problem in $G$.

\item[\textnormal{(2) Directed Spanning Tree ($\HH{$\SF{DST}$}_{\emptyset}$)}] 
\hfill \\
Similar to the previous case, $W_\cc{P}$ is the set of directed 
spanning trees rooted at vertex $1.$

Let $G = (V, E),$ be a directed planar graph such that the 
in-degree and the out-degree for every vertex is at most $2.$
The $\SF{DHAM\!-\!PATH}$ problem in $G$, \ie determining if 
there exists directed Hamiltonian path ending at node $1$ in 
$G$, is $\NP$-hard \cite{GJ79}.
Our goal is to reduce the $\SF{DHAM\!-\!PATH}$ problem in $G$ to 
the $\HH{$\SF{DST}$}_\emptyset$ problem in $G.$

For every vertex $v \in V,$ of in-degree $2$ we consider the edge sets $E_v$ 
where 
$E_v := \{(i,v) \mid (i, v) \in E\}$ with $|E_v| \leq 2$ by our choice of $G.$
In addition, for every vertex $v \in V,$ of out-degree $2$ we consider the 
edge sets $E^v$ 
where 
$E^v := \{(v,i) \mid (v,i) \in E\}$.
With these $E_{v}$s and $E^v$, the $\bigcup$\X{$\SF{DST}$} problem asks if there 
exists a directed spanning tree 
 rooted at vertex $1$ that contains at most
one edge coming in and at most one edge originated from every vertex. 
These constraints restrict the directed spanning tree, $A$ to be a 
$\SF{DHAM\!-\!PATH}$ in $G$, analogously  
to the undirected case.

Hence, the $\NP$-hard $\SF{DHAM\!-\!PATH}$ problem reduces to the 
$\HH{$\SF{DST}$}_{\emptyset}$ problem in $G$.

\item[\textnormal{(3) Undirected Cycle Cover ($\HH{$\SF{UCC}$}_{\emptyset}$)}] 
\hfill \\
Here, $W_\cc{P}$ is the set of {\em undirected cycle covers} on $n$-vertices. 
 
From Hell et al. \cite{Hell88} we know that 
the problem of deciding whether a graph has a $\SF{UCC}$ that 
does not use the cycles of length, (say) $5$ is $\NP$-hard.
We construct an equivalent instance of $\bigcup$\X{$\SF{UCC}$} 
as follows. We choose the edge sets $E_C := \{ e \mid e \in C\}$ 
ranging over every length $5$ cycle $C$ in $G$. Then, a $\SF{UCC}$ 
satisfying the above conditions cannot contain any $5$-cycles. 

Hence, an $\NP$-hard problem reduces to the $\HH{$\SF{UCC}$}_{\emptyset}$ 
problem in $G$.

\item[\textnormal{(4) Directed Cycle Cover ($\HH{$\SF{DCC}$}_{\emptyset}$)}] 
\hfill \\
In this case, $W_\cc{P}$ is the set of {\em directed cycle covers} on 
$n$-vertices. 

The proof follows similar to the undirected case. The 
$\NP$-hard problem we are interested in is determining 
if a  graph has a $\SF{DCC}$ that does not use cycles of length 
$1$ and $2$ \cite{GJ79}. This problem can be expressed 
as $\bigcup$\X{$\SF{DCC}$} by choosing the edge sets 
$E_C := \{ e \mid e \in C\}$ for every length $1$ and 
length $2$ cycle $C$ in $G$.

\item[\textnormal{(5) Bipartite Perfect Matching ($\HH{$\SF{BPM}$}_{\emptyset}$)}] 
\hfill \\
Here, $W_\cc{P}$ is the set of {\em perfect matchings} in a complete 
bipartite graph with $n$-vertices on each side. 

There is a one-to-one correspondence between perfect matchings 
in a bipartite graph $G = (A \cup B, E)$  with $n$ vertices on 
each side and the directed cycle covers in a graph $G' = (V', E')$ 
on $n$ vertices. Every edge $(i, j) \in E'$ corresponds to an 
undirected edge $\{i_A, j_B\} \in E.$ With this correspondence 
the $\HH{$\SF{BPM}$}_{\emptyset}$ problem in $G$ is equivalent 
to the $\HH{$\SF{DCC}$}_{\emptyset}$ problem in $G'.$  Thus, from 
Theorem~\ref{thm:mono-graph}(4) the former becomes $\NP$-hard.
 
\item[\textnormal{(6) Directed Path ($\HH{$\SF{DPATH}$}_{\emptyset}$)}] \hfill \\
We consider $W_\cc{P}$ as the set of {\em directed paths} from $s$ to $t$. 

It is known that given a layout of a directed graph on a plane 
possibly 
containing crossings, the problem of deciding whether there is a 
crossing-free path from $s$ to $t$ is $\NP$-hard \cite{KLN91}. 
This condition can be expressed by picking the each edge set 
$E_i$ as the set of pairs of edges that cross.

\item[\textnormal{(7) Undirected Path ($\HH{$\SF{UPATH}$}_{\emptyset}$)}] \hfill \\
In this case, $W_\cc{P}$ is the set of {\em undirected paths} from $s$ to $t$. 

We can apply the same proof method as the one used for the 
$\HH{$\SF{DPATH}$}_{\emptyset}$ problem on an undirected graph.

\end{description}

This completes the proof for Theorem~\ref{thm:mono-graph}.
\end{proof}



\section{Hidden CSPs with Promise on Instances}
\label{sec:promise}

In this section we consider an extension of the $\HH{$\CSP$}$ framework where the 
instances satisfy some property. For the sake of simplicity, we develop this subject only for the constraint index
revealing model. Formally, let $\SP$ be a $\CSP$, and let $\prom$ be a subset of all instances.
Then $\SP$ with {\em promise} $\prom$ is the $\CSP$  $\SP^{\prom}$ whose instances are only elements of  $\prom$.
One such property is {\em repetition freeness}  where 
the constraints of an instance are pairwise distinct.  We denote by $\rf$ the subset of instances satisfying this property.
For example 
$\SAT{1}^{\rf}$, (as well as $\HH{$\SAT{1}^{\rf}$}$) consists of
pairwise distinct literals. 
Such a requirement is quite natural in the context of certain graph
problems where the constraints are inclusion (or non-inclusion)
of possible edges.
The promise $\HCSP$s framework could also be suitable
for discussing certain graph problems on
special classes of graphs (e.g, connected graphs, planar
graphs, etc.). 

We would like to prove an analogue of the transfer theorem with promise.
Let us be given a promise $\prom$ for the $\CSP$ $\SP$ of type $\cc{R} = \{R_1, \ldots , R_s\}$. 
The corresponding
promise $\bigcup \prom$ for $\bigcup \SP$ is defined quite naturally as follows.
We say that an instance $\cc{C} = (C_1,\ldots,C_m)$ of  $\SP$, where 
$C_j = R_{k_j}(x_{j_1}, \ldots, x_{j_q})$, 
is {\em included} in an instance
$\cc{C'} = (C_1',\ldots,C_m')$ of  $\bigcup \SP$ if 
for every $j=1,\ldots,m$ 
$C_j'=Q_j(x_{j_1},\ldots,x_{j_q})$
for $Q_j\in \bigcup \cc{R}$ such that $R_{k_j}\subseteq Q_j$.
Then $\bigcup \prom$ is defined as the set of instances in $\cc{C'} \in \bigcup \SP$  which 
include $\cc{C} \in \prom$.
In order for the transfer theorem to work, we relax the notion of a solution.
A {\em solution under promise} for $ \cc{C'} \in \bigcup \prom$ has to satisfy two criteria:
it is a satisfying assignment when $\cc{C'}$ includes a satisfiable instance $\cc{C} \in \prom$, and it is $\exc$ when
$\cc{C'}$ is unsatisfiable. However, when all the instances
$\cc{C} \in \prom$ included in $\cc{C'}$
are unsatisfiable but
$\cc{C'}$ is still satisfiable, it can be either a satisfying assignment or $\exc$. We say that an algorithm {\em solves}
$\bigcup \SP^{\bigcup \prom}$ {\em under promise} if $\forall \cc{C'} \in 
\bigcup \prom$, it outputs a solution 
under promise.

Using the above definition in the transfer theorem's proof allows the algorithm for $\HH{S}_{\{\cc{V}\}}$
to terminate, at any moment of time, with the conclusion $\no$ 
as soon as it gets enough information about the instance to exclude satisfiability 
and without making further calls to the revealing oracle.
In some ambiguous cases, it can still call the oracle with an assignment
which satisfies the $\bigcup \SP$-instance. Other cases when the 
satisfiability of a $\bigcup \SP$-instance with promise implies the existence of a satisfiable 
promise-included instance lack this ambiguity.
With these notions the proof of Theorem~\ref{cv_reduction}
goes through and we obtain the following.

\begin{theorem}
 \label{promise_cv_reduction}
Let $\SP^{\prom}$ be a promise $\CSP$.
(a) 
If $\bigcup S^{\bigcup \prom}$ is solvable under promise in time $T$ then $\HH{S}_{\{\cc{V}\}}^{\prom}$  is solvable in time
$O((T + s \times \comp(\cc{R})) \times m \times \min\{\dim (\bigcup \cc{R}), |W_q|\})$. \\
(b) If 
$\HH{S}_{\{\cc{V}\}}^{\prom}$ is solvable in time $T$ 
then $\bigcup S^{\bigcup \prom}$  is solvable under promise in time $O(T \times m \times \comp(\bigcup \cc{R}))$.
\end{theorem}


We apply Theorem~\ref{promise_cv_reduction} to the following problems: 
(1) $\HH{$\SAT{1}$}_{\emptyset}^{\rf}$, 
repetition free $\HH{$\SAT{1}$}$; (2) $\HH{$\SAT{2}$}_{\emptyset}^{\rf}$,
repetition free $\HH{$\SAT{2}$}$; (3) ${\HH{$\TWOCOL$}}_\emptyset^\rf$,
repetition free $\HH{$\TWOCOL$}$; (4) $\HH{$\KWEIGHT$}_\emptyset^\rf$ the 
repetition free hidden version of the following problem. 
The problem $\KWEIGHT$ decides if a $0$-$1$ string
has Hamming weight at least $k$. Formally, we have $w=2$, $q=1$ and 
$\cc{R}=\{\{0\}\}$ and $W$ consists of words of length $\ell$
having Hamming weight $k$. An instance of $\KWEIGHT$
is a collection $(C_1,\ldots,C_m)$ of constraints of the form $x_{i_j}=0$
(formally, $C_j=\{0\}^{i_j}$).
(The string behind these constraints
is $b$ where $b_t=0$ if and only if $t\in \{i_1,\ldots,i_m\}$.)
In a repetition free instance we have $|\{i_1,\ldots,i_m\}|=m$.

Our main motivation for introducing promises on instances is to study 
the effect of prohibiting repetition of constraints.
This requirement potentially makes the hidden problem easier as it will
be indeed the case of $\SAT{1}$ (in the constraint index revealing model). 
However, it neither helps in the case of $\SAT{2}$ nor in the case
of the graph property bipartiteness ($\TWOCOL$). (We remark that
guessing a 2-coloring could be interpreted as covering the complementer
graph by two complete subgraphs and hence could also have been discussed
in the framework of the previous section. Here we show hardness of
the potentially easier, repetition-free version.)
Finally, $\KWEIGHT$ is a simple question where the impact of repetition-freeness
depends on the parameter $k$. We have the following.

\begin{theorem}
\label{rf_statements}
(a) {Repetition free $\HH{$\SAT{1}$}$} with constraint index revealing oracle is easy,
that is $\HH{$\SAT{1}$}_{\emptyset}^{\rf} \in \P$.
(b) $\HH{$\KWEIGHT$}_\emptyset$ is $\NP$-hard for certain $k$, 
but $\HH{$\KWEIGHT$}_\emptyset^\rf\in \P$ for every $k$.
(c) Repetition free $\HH{$\SAT{2}$}$, with constraint index revealing oracle,
that is, $\HH{$\SAT{2}$}_{\emptyset}^{\rf}$ is $\NP$-hard.
(d)Repetition free $\HH{$\TWOCOL$}$, that
is ${\HH{$\TWOCOL$}}_\emptyset^\rf$ is $\NP$-hard.
\end{theorem}

\begin{proof} We prove each part of the theorem separately:

\begin{enumerate}[(a)]

\item We consider every literal as its extended $\ell$-ary relation where $n$ is the number of variables.
This transforms the $\emptyset$-oracle into a $\{\cc{V}\}$-oracle.
A repetition free instance of $\bigcup\SAT{1}$ is 
$\cc{C} = \{C_1,\ldots,C_m\}$, where each $C_j$ is a disjunction of 
literals from $\{x_1,{\overline x_1},\ldots, x_{\ell},{\overline x_{\ell}}\}$
such that there exist $m$ distinct literals $z_1,\ldots,z_m$ with $z_j$ from $C_j$.
A conjunction of literals is satisfiable, if for every $i \in [\ell]$, the literals 
$x_i$ and $\overline x_i$ are not both among them. Hence an algorithm which solves $\HH{$\SAT{1}$}_{\emptyset}^{\rf}$ 
under promise can
proceed as follows. Using a maximum matching algorithm
it selects pairwise different variables $x_{i_1},\ldots,x_{i_m}$ such that $x_{i_j}$ or $\overline x_{i_j}$ is in $C_j$.
If such
a selection is not possible it returns $\exc$. Otherwise
it can trivially find a satisfying assignment.

\item Again, we extend the relations to their $\ell$-ary
counterparts so that the $\emptyset$-oracle is transformed 
into a $\{\cc{V}\}$-oracle. 

An instance of $\bigcup \KWEIGHT$ is
$\cc{C'}=\{C_1',\ldots,C_m'\}$, where there
exist subsets $S_1,\ldots,S_m$ of $[\ell]$
such that the relation for $C_j$ is the set
$\{a\in \ALPHA{w}^{\ell}: a_i=0 \mbox{~for some~}i\in S_j\}$.
Finding a satisfying instance of $\bigcup \cc{R}$
is therefore equivalent to finding a hitting set (a transversal)
of size (at most) $\ell-k$ for the hypergraph $\{S_1,\ldots,S_m\}$.
This problem is $\NP$-hard for, say, $0.01\ell<k<0.99\ell$.

A $\KWEIGHT^\rf$-instance included
in an instance of $\bigcup \KWEIGHT^{\bigcup \rf}$
corresponding to subsets $S_1,\ldots,S_m$ 
consists of constraints $x_{i_j}\neq 0$ for
$m$ different indices $i_1,\ldots,i_m$
with $i_j\in S_j$. Obviously, such
a set of constraints is satisfiable by an
element of $W$ if and only if $m\leq \ell-k$. These
observations immediately give the following efficient solution
under promise for $\bigcup \KWEIGHT^{\bigcup \rf}$. 
If $m>\ell-k$ we return \exc. Otherwise, using a maximum
matching algorithm we find $m$ different places $i_1,\ldots,i_m$
with $i_j\in S_j$ (which must exist by the promise)
and return an assignment from $W$ which can be found 
in an obvious way. 

\item To work in the framework
of a $\{\cc{V}\}$-oracle rather than a $\emptyset$-oracle,
we consider every clause as its extended $n$-ary relation where 
$n$ is the number of variables.
This transforms the $\emptyset$-oracle into a $\{\cc{V}\}$-oracle.
We reduce $\SAT{3}$ to $\bigcup \SAT{2}^{\bigcup \rf}$
as follows. Let $\phi=\bigwedge_{j=1}^m C_j$
be a $3\!-\!\SF{CNF}$ where 
$$C_j= x_{j_1}^{b_1}\vee x_{j_2}^{b_2} \vee x_{j_3}^{b_3}.$$
(Here $b_i \in \{0,1\}$ and $x^1$ denotes $x$, $x^0$ stands for $\overline x$.)
For each $j = 1,\ldots,m$ we introduce a new variable $y_j$.
We will have $2m$ new clauses: 
$$
C_j'=x_{j_1}^{b_1}\vee x_{j_2}^{b_2} \vee x_{j_3}^{b_3} \vee y^{0}_j
\mbox{~~and~~}
C_j''=x_{j_1}^{b_1}\vee x_{j_2}^{b_2} \vee x_{j_3}^{b_3}
\vee y^{1}_j$$
for each $j$.
Put $\phi'=\bigwedge_{j=1}^m(C_j'\wedge C_j'').$
Then $\phi$ is satisfiable if and only if $\phi'$ is satisfiable.
In fact, there is a $1$ to 	$2^m$ correspondence
between the
assignments satisfying $\phi$ and those satisfying $\phi'$:
only the values assigned to the first $\ell$ variables matter.
Also, the included constraints $(x_{j_1}^{b_1}\vee y^{1}_j)$
and $(x_{j_1}^{b_1}\vee y^{0}_j)$ for all $j=1,\ldots,m$ form
a system of $2m$ different $2\!-\!\SF{CNF}$s. Furthermore, if $\phi'$ is satisfied by
an assignment then we can select a satisfiable system of $2m$ pairwise distinct
sub-constraints: for each $j$ we pick $s\in\{1,2,3\}$ such that
$x_{j_s}^{b_s}$ is evaluated to $1$ and take
$(x_{j_s}^{b_s}\vee y^{1}_j)$ and
$(x_{j_s}^{b_s}\vee y^{0}_j)$ for $j=1,\ldots,m$.

\item Here the alphabet is $\ALPHA{2}=\{0,1\}$, $q=2$, $\cc{R}$
has one element ``$\neq$", that is $\{(1,0),(0,1)\}$.
An instance of $\TWOCOL^\rf$ consists of a set of constraints
of the form $x_u\neq x_v$ for $m$ pairwise distinct unordered
pairs $\{u,v\}$ from $\{1,\ldots,\ell\}$  (corresponding to
the edges of a graph). (Here we again work in the context of 
the extensions of the relation ``$\neq$" to arity $\ell=n$.)

An instance of $\bigcup \TWOCOL$ is a collection
$\{C_1,\ldots,C_m\}$, where each $C_j$ is
a disjunction of constraints of the form $x_u\neq x_v$.
In an equivalent view, an instance of $\bigcup \TWOCOL$ can be described
by the collection of edge sets (graphs)
$E_1,\ldots,E_m$  on vertex set $[n]$ 
and a satisfying assignment can be described
by a coloring $c:\{1,\ldots,n\}\rightarrow \{0,1\}$ such that for every $j$ 
there exists an edge $e_j\in E_j$ with endpoints having different colors. 
It is clear that if the edge sets $E_1,\ldots,E_m$ are disjoint then
the instance is repetition free and the solutions under promise 
coincide with the solutions in the normal sense.

Let $E_1,\ldots,E_m$ be edge sets describing an instance of 
$\bigcup \TWOCOL$. Put $s_j=|E_j|$. For each $j$ we introduce $2s_j$ 
new vertices: $uvj1,uvj2$ for each $\{u,v\}\in E_j$,
$2s_j$ new one-element edge sets
$E_{uvj1}=\{\{u,uvj1\}\}$ and
$E_{uvj2}=\{\{v,uvj2\}\}$; while 
$E_j$ is replaced with an edge $E_j'$ set consisting of
$s_j$ edges: 
$\{uvj1,uvj2\}$ for each $\{u,v\}\in E_j$. 
It turns out that the $\bigcup \TWOCOL$ 
problem on the $n+2\sum_{j=1}^m s_j$ vertices
with the new $m+2\sum_{j=1}^m s_j$ edge sets
is equivalent to the original one and
solutions of the two problems can be easily (and efficiently)
mapped to each other. The new edge sets are
pairwise disjoint and hence the repetition free version
of the new $\bigcup \TWOCOL$ problem is the same as the
the non-promise version. 

Theorem~\ref{c_boring} shows that 
non-promise $\bigcup{\TWOCOL}$ is $\NP$-hard.
By the reduction above, so is its repetition free
version. \qedhere
\end{enumerate}
\end{proof}

\paragraph{On group isomorphism.} Isomorphism of a hidden multiplication table with
a given group, a problem discussed in \cite{BCZ13},
can also be cast in the framework of promise $\HCSP$s. We consider
the following problem $\GROUPEQ$ (equality with a group from a class).
Let $\cc{G}$ be a family of groups on the set $[k]$,
that is, a set of multiplication tables
on $[k]$ such that each multiplication table defines
a group. The task is to decide whether a hidden group
structure $b(\,,\,)$ is equal to some $a(\,,\,)$ from $\cc{G}$
and if yes, find such an $a(\,,\,)$. (Note that a solution of the latter 
task will give the whole table for $b(\,,\,)$.)

 
We define $\GROUPEQ(\cc{G})$ as a promise $\CSP$ as follows.
First we consider the $\CSP$ $\ENT(\cc{G})$ with the following 
parameters and type. We have $w=k$, $W=\cc{G}$, 
$\cc{R}=\{\{w\}:w\in [k]\}$,
$\ell=k^2$. It will be convenient
to consider assignments as $k\times k$ tables
with entries from $[k]$, that is, functions $[k]^2\rightarrow [k]$.
(Implicitly, we use a bijection between the index set
$\{1,\ldots,\ell\}$ and $[k]^2$.)
The number of constraints is $m=k^2$ and an 
instance is a collection  
$x_{(u_h,v_h)}=b_h$, where $h=1,\ldots,m$, and $u_h, v_h, b_h\in[k]$.
Thus the assignment satisfying $\ENT(\cc{G})$ are $k\times k$ 
multiplication tables from $\cc{G}$
which have prescribed values at $k^2$ (not necessarily distinct)
places.

We say that an instance for $\ENT(\cc{G})$ belongs to the 
promise $\GROUP$ if two conditions are satisfied.
Firstly, there is one constraint for
the value taken by each place. Formally,
the map $\tau:h\mapsto (u_h,v_h)$ is a bijection 
between $\{1,\ldots,m\}$ and $[k]^2$. As a consequence,
by setting $b(u,v):=b_{\tau^{-1}(u,v)}$, we have a constraint
$x_{u,v}=b(u,v)$ for pair $(u,v)\in [k]^2$. 
The second -- essential -- condition is that
the multiplication given by $b(\,,\,)$
defines a group structure on $[k]$. 
The promise problem $\GROUPEQ(\cc{G})$ is the problem
$\ENT(\cc{G})^{\GROUP}$. 

We consider the promise problem
$\HH{$\ENT(\cc{G})$}_{\{V\}}^\GROUP$ 
(which we denote by $\HH{$\GROUPEQ(\cc{G})$}$ for short)
and the corresponding problem
$\bigcup \ENT(\cc{G})^{\bigcup \GROUP}_{\{V\}}$
(short notation: $\bigcup \GROUPEQ(\cc{G})$).
In this $\HCSP$, if $a(\,,\,)$ is different from
$b(\,,\,)$, the oracle reveals
a pair $(u,v)$ such that $a(u,v)\neq b(u,v)$.
 
We note here that the case of $\HH{$\GROUPEQ(\cc{G})$}$
where $\cc{G}$ consists of all isomorphic copies 
of a group $G$ in fact covers the problem of finding
an isomorphism with $G$ discussed in \cite{BCZ13}. For that problem, 
the input to the verification oracle is a bijection 
$\phi:[k]\rightarrow G$. Recall that $b(\,,\,)$ encodes the hidden group 
structure, and we assume $G$ is specified by the binary relation $g(\,,\,)$. Then, 
in the case when $\phi$ is not an isomorphism, the oracle has to
reveal $u,v\in [k]$ such that, the product $g(\phi(u),\phi(v))$ 
does not equal $\phi(b(u,v))$ in $G$. Indeed, given
$\phi$ we can define (and even compute) the multiplication
$a_{\phi}(\,,\,)$ on $[k]$ -- by taking $a_\phi(x,y)=
\phi^{-1}(g(\phi(x),\phi(y))$ -- such that $\phi$ becomes
an isomorphism from the group given by $a_{\phi}(\,,\,)$
to $G$. Then $\phi$ is an isomorphism from the group
given by $b(\,,\,)$ if and only if 
$a_{\phi}(\,,\,)=b(\,,\,)$. Furthermore, if it is not the case
then the oracle given in \cite{BCZ13} reveals a pair $(u,v)$ such that
$a_{\phi}(u,v)=b(u,v)$, exactly what is expected from
a revealing oracle for $\HH{$\GROUPEQ(\cc{G})$}$. 

An instance of 
$\bigcup \ENT(\cc{G})$
consists of $k^2$ constraints
expressing that $a(u_h,v_h)\in S_h$ where $S_h\in
2^{[k]}\setminus \emptyset$ for $h\in[m]=[k^2]$.
An instance of the promise version 
$\bigcup \GROUPEQ(\cc{G})$ 
(which is equal to $\bigcup \ENT(\cc{G})^{\bigcup \GROUP}_{\{V\}}$)
should satisfy that
$\{(u_h,v_h):h=1,\ldots,m\}=[k]^2$, that is, our constraints
are actually $x_{(u,v)}\in S(u,v)$ for
a map $S(\,,\,):[k]^2\rightarrow 2^{[k]}$. Furthermore,
there is a map $b(\,,\,):[k]^2\rightarrow [k]$ with $b(u,v)\in S(u,v)$
for every $(u,v)\in [k]^2$ such that $b(\,,\,)$ gives a group
structure.

Now we are ready to reprove Theorem 11 in \cite{BCZ13}. Note that our proof is 
considerably shorter than the original proof. 

\begin{theorem}
The problem $\HH{$\GROUPEQ(\cc{G})$}$ is $\NP$-hard.
\end{theorem}
\begin{proof}
Let $p$ be a prime. We show that finding Hamiltonian cycles in Hamiltonian digraphs of size $p$ 
is reducible in polynomial time to $\HH{$\GROUPEQ(\cc{G})$}$. The former problem is $\NP$-hard,
since an algorithm that in polynomial time finds a Hamiltonian cycle in a Hamiltonian digraph obviously
can decide if an arbitrary digraph $G$ has a Hamiltonian cycle: it just runs on $G$ and then tests if the outcome is indeed
a Hamiltonian cycle.

Choose $\cc{G}$ as the set of all
group structures on $[p]$. As every group having $p$
elements is isomorphic to $Z_p$,  $\cc{G}$ coincides
with the group structures on $[p]$ isomorphic to $Z_p$. 
We essentially translate the arguments 
given in \cite{BCZ13} to the setting of $\bigcup \GROUPEQ$ as follows. This 
suffices to prove the statement due to the transfer theorem stated in Theorem~\ref{promise_cv_reduction}.

Let $([p],E)$ be a Hamiltonian directed graph (without loops) on $[p]$. 
Fix $z\in [p]$. For $u\in [p]$, let $S(u,z)=\{v:(u,v)\in E\}$, and 
$S(u,v)=[p]$ for $v\neq z$. Let $\phi:[p]\rightarrow \{0,\ldots,p-1\}=Z_p$ 
be a bijection that defines a Hamiltonian cycle in $([p],E)$. Then
$b(x,y)=a_\phi(x,y):=\phi^{-1}(\phi(x)+\phi(y))$ gives a group structure
on $[p]$ (isomorphic to $Z_p$ via $\phi$) consistent
with the constraints given by $S(\;,\;)$.
Conversely, if $b(\,,\,)$ gives a group structure 
(necessarily isomorphic to $Z_p$) consistent
with $S(\;,\;)$ then the pairs $(u,b(u,z))$ $(u\in [p])$
form a Hamiltonian cycle in $([p],E)$. Thus finding
Hamiltonian cycles in Hamiltonian digraphs
on $p$ points can be reduced to $\bigcup \GROUPEQ$ 
on $p$ elements.
\end{proof}
As an example, suppose we have $p=3$, and the edges are $2\to 1, 1\to 3, 3\to 2$.  
Then set 
$z=2$, so $\phi(2)=1$, $\phi(1)=2$, $\phi(3)=0$. (That is, $i$ is the $\phi(i)$th 
vertex to be visited in this Hamiltonian cycle, where $\phi(i)$ should be 
understood as modulo $p$.) It can be verified that $b(x, 
2)\in S(x, 2)=\{y: (x, y)\in E\}$. On the other hand, if we set $b(x, y)$ to be 
isomorphic to $Z_p$ by the correspondence just given by $\phi$, then the path $(u, 
b(u, 2))$ forms a Hamiltonian cycle. 









\medskip

\noindent{\bf Acknowledgements.}
The authors are grateful to the anonymous reviewers for their helpful
remarks and suggestions. 
This research  was  supported in part by 
the Hungarian National Research, Development and Innovation Office -- NKFIH
(Grant K115288), the Singapore Ministry of Education and the 
National Research Foundation, also through the Tier 3 Grant 
``Random numbers from quantum processes" MOE2012-T3-1-009,
by the European Commission
IST STREP project Quantum Algorithms (QALGO) 600700, the French ANR Blanc Program 
under contract ANR-12-BS02-005 (RDAM project), and Australian Research Council 
DECRA DE150100720.
\bibliographystyle{alpha}
\bibliography{ref}

\newcommand{\etalchar}[1]{$^{#1}$}
\begin{thebibliography}{HKKK88}

\bibitem[BCD{\etalchar{+}}13]{BCD+13}
Xiaohui Bei, Ning Chen, Liyu Dou, Xiangru Huang, and Ruixin Qiang.
\newblock Trial and error in influential social networks.
\newblock In Inderjit~S. Dhillon, Yehuda Koren, Rayid Ghani, Ted~E. Senator,
  Paul Bradley, Rajesh Parekh, Jingrui He, Robert~L. Grossman, and Ramasamy
  Uthurusamy, editors, {\em KDD}, pages 1016--1024. ACM, 2013.

\bibitem[BCZ12]{BCZ12}
Xiaohui Bei, Ning Chen, and Shengyu Zhang.
\newblock On the complexity of trial and error.
\newblock {\em CoRR}, abs/1205.1183, 2012.

\bibitem[BCZ13]{BCZ13}
Xiaohui Bei, Ning Chen, and Shengyu Zhang.
\newblock On the complexity of trial and error.
\newblock In Dan Boneh, Tim Roughgarden, and Joan Feigenbaum, editors, {\em
  STOC}, pages 31--40. ACM, 2013.

\bibitem[BCZ15]{BCZ13b}
Xiaohui Bei, Ning Chen, and Shengyu Zhang.
\newblock Solving linear programming with constraints unknown.
\newblock In Magn\'us~M. Halld\'orsson, Kazuo Iwama, Naoki Kobayashi, and
  Bettina Speckmann, editors, {\em ICALP}, pages 129--142. LNCS, 2015.

\bibitem[GJ79]{GJ79}
M.~R. Garey and David~S. Johnson.
\newblock {\em Computers and Intractability: A Guide to the Theory of
  NP-Completeness}.
\newblock W. H. Freeman, 1979.

\bibitem[HKKK88]{Hell88}
Pavol Hell, David~G. Kirkpatrick, Jan Kratochv\'{\i}l, and Igor Kr\'{\i}z.
\newblock On restricted two-factors.
\newblock {\em SIAM J. Discrete Math.}, 1(4):472--484, 1988.

\bibitem[IKQ{\etalchar{+}}14]{conf_version}
G{\'{a}}bor Ivanyos, Raghav Kulkarni, Youming Qiao, Miklos Santha, and Aarthi
  Sundaram.
\newblock On the complexity of trial and error for constraint satisfaction
  problems.
\newblock In {\em Automata, Languages, and Programming - 41st International
  Colloquium, {ICALP} 2014, Copenhagen, Denmark, July 8-11, 2014, Proceedings,
  Part {I}}, pages 663--675, 2014.

\bibitem[Kho02]{Kho02}
Subhash Khot.
\newblock On the power of unique 2-prover 1-round games.
\newblock In {\em Proceedings on 34th Annual {ACM} Symposium on Theory of
  Computing, May 19-21, 2002, Montr{\'{e}}al, Qu{\'{e}}bec, Canada}, pages
  767--775, 2002.

\bibitem[KLN91]{KLN91}
Jan Kratochv\'{\i}l, Anna Lubiw, and Jaroslav Ne\v{s}et\v{r}il.
\newblock Noncrossing subgraphs in topological layouts.
\newblock {\em SIAM J. Discret. Math.}, 4(2):223--244, March 1991.

\bibitem[Sch78]{Sch78}
Thomas~J. Schaefer.
\newblock The complexity of satisfiability problems.
\newblock In Richard~J. Lipton, Walter~A. Burkhard, Walter~J. Savitch, Emily~P.
  Friedman, and Alfred~V. Aho, editors, {\em STOC}, pages 216--226. ACM, 1978.

\end{thebibliography}
\newpage
\appendix

\end{document}